\documentclass[letterpaper,onecolumn,10pt,accepted=2022-11-01]{quantumarticle}
\pdfoutput=1
\usepackage[utf8]{inputenc}
\usepackage[english]{babel}
\usepackage[T1]{fontenc}

\usepackage{amsmath,amsthm,amssymb}
\usepackage{graphicx}
\usepackage{color}
\usepackage{array}
\usepackage{braket}
\usepackage{bbm}

\usepackage[hidelinks]{hyperref}
\hypersetup{
    colorlinks=true
}
\usepackage{cleveref}

\usepackage{thm-restate}
\theoremstyle{plain}
\newtheorem{theorem}{Theorem}

\newtheorem{lemma}[theorem]{Lemma}
\newtheorem{proposition}[theorem]{Proposition}
\newtheorem{corollary}[theorem]{Corollary}

\newtheorem{problem}[theorem]{Problem}

\theoremstyle{definition}
\newtheorem{definition}[theorem]{Definition}

\Crefformat{equation}{Eq.~(#2#1#3)} 
\Crefformat{lemma}{Lemma~#2#1#3}
\Crefformat{corollary}{Corollary~#2#1#3}
\Crefformat{section}{Section~#2#1#3}
\Crefformat{appendix}{Appendix~#2#1#3}
\Crefformat{figure}{Fig.~#2#1#3}
\Crefrangeformat{equation}{Eqs.~#3(#1)#4--#5(#2)#6}

\DeclareMathOperator{\id}{\mathbbm{1}}
\newcommand{\beq}{\begin{equation}}
\newcommand{\eeq}{\end{equation}}
\newcommand{\beqs}{\begin{equation*}}
\newcommand{\eeqs}{\end{equation*}}

\newcommand{\GPEa}{\operatorname{GPE}_\pm}
\newcommand{\RGPEa}{\operatorname{RGPE}_\pm}

\newcommand{\GPEpa}{\left|\operatorname{GPE}\right|_{\pm}^2}
\newcommand{\RGPEpa}{\left|\operatorname{RGPE}\right|_{\pm}^2}

\newcommand{\FBPP}{\mathsf{FBPP}}
\newcommand{\NP}{\mathsf{NP}}
\newcommand{\BPP}{\mathsf{BPP}}
\newcommand{\PH}{\mathsf{PH}}
\renewcommand{\P}{\mathsf{P}}
\newcommand{\poly}{\mathsf{poly}}
\newcommand{\G}{\mathcal{G}}
\newcommand{\Ord}{\textrm{O}}

\newcommand{\norm}{\mathcal{Z}}

\DeclareMathOperator\arctanh{arctanh}
\DeclareMathOperator\Per{Per}
\DeclareMathOperator\Haf{Haf}
\DeclareMathOperator\Tr{Tr}
\DeclareMathOperator\Var{Var}

\newcommand{\seq}[1]{\begin{subequations}#1\end{subequations}}

\newcommand{\eq}[1]{\begin{align}#1\end{align}}

\begin{document}
\title{The Complexity of Bipartite Gaussian Boson Sampling}
\author{Daniel Grier}
\affiliation{Institute for Quantum Computing, University of Waterloo, Canada}
\affiliation{Department of Computer Science and Engineering and Department of Mathematics, University of California, San Diego, US}
\author{Daniel J. Brod}
\affiliation{Instituto de F\'isica, Universidade Federal Fluminense,
 Niter\'oi, RJ, 24210-340, Brazil}
\author{Juan Miguel Arrazola}
\affiliation{Xanadu, Toronto, ON, M5G 2C8, Canada}
\author{Marcos Benicio de Andrade Alonso}
\affiliation{Instituto de F\'isica, Universidade Federal Fluminense,
 Niter\'oi, RJ, 24210-340, Brazil}
\author{Nicol\'as Quesada}
\affiliation{Department of Engineering Physics, \'Ecole Polytechnique de Montr\'eal, Montr\'eal, QC, H3T 1JK, Canada}

\begin{abstract}
Gaussian boson sampling is a model of photonic quantum computing that has attracted attention as a platform for building quantum devices capable of performing tasks that are out of reach for classical devices. There is therefore significant interest, from the perspective of computational complexity theory, in solidifying the mathematical foundation for the hardness of simulating these devices. We show that, under the standard Anti-Concentration and Permanent-of-Gaussians conjectures, there is no efficient classical algorithm to sample from ideal Gaussian boson sampling distributions (even approximately) unless the polynomial hierarchy collapses. The hardness proof holds in the regime where the number of modes scales quadratically with the number of photons, a setting in which hardness was widely believed to hold but that nevertheless had no definitive proof.
    
Crucial to the proof is a new method for programming a Gaussian boson sampling device so that the output probabilities are proportional to the permanents of submatrices of an arbitrary matrix.  This technique is a generalization of Scattershot BosonSampling that we call BipartiteGBS.  We also make progress towards the goal of proving hardness in the regime where there are fewer than quadratically more modes than photons (i.e., the high-collision regime) by showing that the ability to approximate permanents of matrices with repeated rows/columns confers the ability to approximate permanents of matrices with no repetitions.  The reduction suffices to prove that GBS is hard in the constant-collision regime.  
\end{abstract}
\maketitle


\section{Introduction}\label{sec:intro}

The quest for quantum computational advantage has given rise to a surprisingly fruitful relationship between computer science and physics: theorems provide a foundation for experiments, while practical considerations set challenges for new mathematics. Consider for instance the role of BosonSampling~\cite{aaronson2013}. It gave strong complexity-theoretic evidence that even a weak photonic device could perform a task that is classically intractable. This work motivated future experimental demonstrations~\cite{tillmann2013, spring2013boson, crespi2013integrated, broome2013photonic} that inspired new theoretical models~\cite{lund2014, hamilton2017}, which in turn resulted in further experiments~\cite{bentivegna2015experimental,wang2017high,zhong2019experimental}.

Gaussian boson sampling (GBS) is a paradigm where a Gaussian state of light is prepared, then measured in the photon-number basis~\cite{hamilton2017, kruse2019detailed,bromley2020applications}. This approach offers several benefits. First, Gaussian states can be prepared by a combination of squeezing, displacement, and linear interferometers, which can in principle be applied deterministically and implemented with nanophotonic integrated circuits~\cite{arrazola2021quantum}. This means they can potentially be mass-produced and scaled rapidly~\cite{wang2019integrated, vernon2018scalable}. Moreover, the inclusion of squeezing and displacements allows more versatility in programming GBS devices, a property that is leveraged in several GBS-based algorithms~\cite{huh2015boson, arrazola2018using, banchi2020molecular, jahangiri2020point, schuld2020measuring, jahangiri2020quantum, banchi2020training}.

GBS has already been used to claim quantum computational advantage~\cite{madsen2022quantum,zhong2020quantum, zhong2021phase}, and there are several more proposals for hard-to-simulate GBS experiments~\cite{lund2014, hamilton2017, deshpande2021quantum}.  That said, we believe these experiments reveal that significant progress is still required to bridge the gap between our theoretical hardness arguments and what is currently achievable in the lab.  For example, state-of-art GBS experiments consisted of 216 modes with up to 125 photons~\cite{madsen2022quantum} and  144 modes and at most 113 clicks \cite{zhong2021phase}, but all hardness arguments currently assume that the number of modes is at least quadratic in the number of photons, and sometimes worse.\footnote{This mismatch between theory and experiment occurs because of photon loss in the interferometer~\cite{garcia2019simulating,qi2020regimes}. For experiments requiring Haar random interferometers, each mode must be able to exchange light with every other mode since with probability one every entry of a random Haar unitary matrix is nonzero. In particular, since random $m$-mode interformeters are built from $\Theta(m^2)$ 2-mode beamsplitters \cite{reck1994,clements2016,de2018simple,bell2021further}, this implies that the depth of the circuit implementing the interferometer is proportional to the width. If there is a fixed transmission per beamsplitter layer $\eta\leq 1$ the total transmission of the interferometer will decay exponentially as $\eta^{m}$. It is then clear that it is more desirable to have $m$ scaling with $n$ and not $n^2$ to have a significant fraction of the photons arriving into the detectors.}

Furthermore, the underlying physics of GBS is such that the probability of a given output state is described by the hafnian, a matrix function similar, but not identical to, the permanent. Because of this, new conjectures tailored to this paradigm are sometimes required---see, for example, the Hafnian-of-Gaussians conjecture in \cite{hamilton2017} which parallels the Permanent-of-Gaussians conjecture in \cite{aaronson2013}.  While the goal of this paper is not to compare the merits of the individual conjectures, we do feel that having fewer standard conjectures is generally preferable.

This paper introduces \emph{Bipartite Gaussian Boson Sampling}\footnote{``Bipartite'' refers to the Husimi covariance matrix of our output states, which is characterized by a bipartite adjacency matrix.} (BipartiteGBS) as a new method for programming a GBS device which will begin to address some of these challenges.  The key property of BipartiteGBS experiments is that the output probabilities are proportional to the permanents of submatrices of \emph{arbitrary} matrices.  Contrasted with traditional BosonSampling, where the output probabilities are dictated by permanents of submatrices of a unitary matrix, BipartiteGBS can be seen as a powerful new tool on which to build hardness-of-simulation arguments. In particular, this paper will focus on the hardness of approximately sampling from BipartiteGBS distributions with a classical device.

As it turns out, our construction is a strict generalization of Scattershot BosonSampling~\cite{lund2014}, a different GBS setup for which the output probabilities are given by permanents of unitary matrices.  Because of this, the computational hardness of Scattershot BosonSampling can be rooted in the same conjectures on which the hardness of BosonSampling is based---namely, that Gaussian permanent estimation is $\#\P$-hard and that Gaussian permanents anti-concentrate. However, this also means that Scattershot BosonSampling inherits the same technical caveats of BosonSampling.  In particular, to guarantee that the $n \times n$ submatrices of an $m \times m$ unitary appear Gaussian, Aaronson and Arkhipov require that $m = \omega(n^5)$.  Therefore, all their hardness proofs are technically within this regime.  To be clear, it is widely assumed that $m = \omega(n^2)$ is sufficient, but to the authors' knowledge, there is currently no definitive proof.\footnote{Perhaps the closest result is that of Jiang \cite{jiang2006many} who shows that the $n \times n$ submatrices of $m \times m$ real orthogonal matrices converge in total variation to real Gaussian matrices whenever $m = \omega(n^2)$.  Unfortunately, Jiang does not bound the \emph{rate} of this convergence, which is required by the BosonSampling hardness arguments. Moreover (though perhaps less importantly), the BosonSampling arguments are based on the submatrices of complex unitary matrices rather than real orthogonal ones.  Finally, Jiang proves that $m = O(n^2)$ does not suffice, but we note that the results of this paper will hold even when $m = \Theta(n^2)$.}  Even more general Gaussian boson sampling protocols suffer from this problem, and recent approximate average-case hardness proofs for GBS must conjecture directly that quadratically-many modes suffice~\cite{deshpande2021quantum}.

The main technical contribution of this paper is to show that BipartiteGBS can be used to close this loophole.  That is, the hardness of GBS \emph{can} be based on the exact same set of conjectures as BosonSampling, while also working in the regime where the number of modes $m$ is quadratic in the expected number of photons $\braket{n}$.\footnote{Unlike BosonSampling where the number of photons is fixed, the number of photons $n$ in a GBS experiment is itself a random variable that is based on the squeezing parameters of the system.  Therefore, we cannot guarantee that the number of modes $m$ is actually quadratic in the number of photons $n$ since there may be some small probability for which there are many more/fewer photons.  Instead, what we can say is that the number of modes is quadratic in the \emph{expected} number of photons $\braket{n}$.  Coupled with bounds on the variance of $n$, we can conclude that the output distribution is dominated by output states which obey the $m = \Theta(n^2)$ relationship.}  Formally, we prove the following theorem:
\begin{theorem}
\label{thm:intro_main}
Suppose there is a classical oracle $\mathcal O$ which approximately samples from a BipartiteGBS distribution with $m = \Theta(\braket{n}^2)$.  Then, $\#\P \subseteq \FBPP^{\NP^\mathcal{O}}$ assuming the Permanent-of-Gaussians Conjecture and the Permanent Anti-Concentration Conjecture.
\end{theorem}

In other words, assuming the BosonSampling conjectures, there is no efficient classical algorithm for BipartiteGBS in this regime unless there is a collapse in the polynomial hierarchy that complexity theorists consider to be extremely unlikely.  Unsurprisingly, the proof of this theorem will leverage the fact that the output probabilities of BipartiteGBS experiments are based on permanents of arbitrary matrices.  Since the hardness of BosonSampling is based on Gaussian permanents anyway, an obvious-in-retrospect idea is to simply \emph{start} with those Gaussian matrices.  Note that, as in all photonic experiments, the probabilities are actually governed by submatrices of the transition matrix.  Clearly, however, the submatrices of a Gaussian matrix---i.e., a matrix for which each entry is an i.i.d.\ complex Gaussian number---are also Gaussian.  This trivializes the ``Hiding Lemma'' often required in other hardness arguments.\footnote{A notable exception is the proposal of Kruse et al.\ \cite{kruse2019detailed} which also circumvents traditional hiding arguments. Their approach is similar---program a GBS device with outputs proportional to hafnians of arbitrary symmetric matrices, which trivializes hiding symmetric matrices. While they conjecture a hardness argument in the regime where the number of photons is linear in the number of modes, our paper shows that the types of challenges that will be required to carry out the full complexity-theoretic argument.}

That said, the proof of \Cref{thm:intro_main} is not itself trivial.  In particular, we must prove that the BipartiteGBS experiments that we ask the classical oracle to simulate have sufficient probability mass on those outcomes which correspond to $\#\P$-hard permanents.  To do so, we bound important normalization factors in the output distribution using results from random matrix theory on ensembles of random Gaussian matrices.  While our proof unsurprisingly borrows many ideas from the original BosonSampling hardness argument~\cite{aaronson2013}, it can be entirely understood without direct reference to it, and we suspect that many will find our rigorous hardness proofs of GBS to be preferable to some in the existing literature.

To complement our formal analysis of the normalization factors (which are sufficient to obtain \Cref{thm:intro_main}), we present a more heuristic analysis of other aspects of our experiment that might be relevant to experimentalists hoping to claim quantum advantage (most likely with a more speculative set of conjectures).  In particular, we give formulas for the expectation and variance of the number of ``clicks'' in the output distribution---that is, the number of modes that contain at least one photon when measured.  Intuitively, this is an important quantity for hardness arguments since the probability of any particular outcome is proportional to a permanent of a matrix whose rank is equal to the number of clicks in that outcome.  Since classical intractability is tied to the complexity of these permanents, and we know that permanents of low-rank matrices have efficient classical algorithms \cite{barvinok1996two}, we would like to avoid distributions with low click rates.  Thankfully, we prove that this is generally not the case by showing that the expected number of clicks in a BipartiteGBS experiment with Gaussian transition matrices is the harmonic mean of the number of modes and expected number of photons:
\beq\label{eq:harmonic_mean}
\frac{2}{\frac{1}{\mathbb E[\braket{n}]} + \frac{1}{m}}.
\eeq
So, for example, even when we expect that there are only linearly-many more modes than photons, we have that the expected number of clicks is linear.  This formula is obtained by assuming that the singular values of a Gaussian matrix are drawn independently and exactly from the quarter-circle law.\footnote{It is only known that the entire distribution of singular values approaches the quarter-circle law.  See \Cref{subsec:gaussian} for more discussion.}  We present numerical evidence (see \Cref{fig:comparison}) showing that these formulas accurately predict the click distributions of random BipartiteGBS experiments.

Finally, we make preliminary progress towards a GBS hardness result in the regime where there are fewer than quadratically many more modes than photons.\footnote{It is worth noting that for the task of \emph{exact} sampling, it was already known that BosonSampling is hard in the $m = o(n^2)$ regime, and in fact, even if $m=n$. Specifically, Grier and Schaeffer give a BosonSampling sampling experiment in the $m=n$ regime for which a particular output probability is $\#\mathsf{P}$-hard to approximate \cite{gs:unitary_permanent}. Combined with the Stockmeyer counting arguments of \cite{aaronson2013}, this shows classical intractability of the exact sampling task predicated on the non-collapse of the polynomial hierarchy.}  In this regime, we can no longer guarantee that there is at most one photon per mode.  These photon collisions imply that the output probabilities are no longer described by permanents of simple submatrices, but rather by submatrices which have some rows and columns repeated.  To this end, we define a new ``Permanent-of-Repeated-Gaussians'' problem for which the goal is to approximate such permanents.  We provide numerics suggesting that a classical simulation of GBS in the high-collision regime can be leveraged to solve the Permanent-of-Repeated-Gaussians problem based on a Stockmeyer counting argument similar to that in \Cref{thm:intro_main}.  In other words, we could show approximate average-case hardness for GBS in the $m = o(\braket{n}^2)$ regime if we made the following assumptions: the $\#\P$-hardness of the Permanent-of-Repeated-Gaussians problem; a plausible conjecture in random matrix theory.  We caution that these assumptions remain relatively unexplored.

As a first step towards understanding the Permanent-of-Repeated-Gaussians, we show that there is some sense in which we can reduce arbitrary matrix permanents to permanents of matrices with repeated rows and columns:

\begin{theorem}
\label{thm:reduction_intro}
Given an oracle $\mathcal O$ that can approximate $\Per(B)$ for any matrix $B \in \mathbb C^{n \times n}$ that has $k$ row/column repetitions to additive error $\epsilon$, there is an $\FBPP^{\mathcal O}$ algorithm that can approximate $\Per(A)$ for arbitrary matrices $A \in \mathbb C^{(n-k) \times (n-k)}$ to additive accuracy $O(\epsilon / 1.498^k)$.
\end{theorem}

The reduction in this theorem has the additional nice property that if the matrix $A$ is Gaussian, then the oracle is only queried on matrices $B$ that are also Gaussian.  This makes it an ideal candidate for use in a hardness reduction because we can only require a classical simulator to sample from distributions that we can sample from quantumly.  Unfortunately, while we are getting an exponential improvement in the accuracy to $\Per(A)$, we show that it is still insufficient to conclude the $\#\P$-hardness of Permanent-of-Repeated-Gaussians from the usual Permanent-of-Gaussians conjecture.  That said, if we assume that there are only constantly-many collisions we \emph{can} show such a reduction.  Furthermore, in this regime we can work directly with the magnitude of the permanent, avoiding the need for an additional anti-concentration conjecture.
\begin{theorem}
\label{thm:reduction_constant_collisions}
Given an oracle $\mathcal O$ that can approximate $|\Per(B)|^2$ for any matrix $B \in \mathbb C^{n \times n}$ that has $k$ row/column repetitions to additive error $\epsilon$, there is an $\FBPP^{\mathcal O}$ algorithm that can approximate $|\Per(A)|^2$ for arbitrary matrices $A \in \mathbb C^{(n-k) \times (n-k)}$ to additive accuracy which is exponential in $k$, but polynomial in $n$ and $\epsilon$.
\end{theorem}

This theorem follows almost immediately by combining \Cref{thm:reduction_intro} with the polynomial interpolation techniques used to prove the classical hardness of BosonSampling with constantly-many lost photons \cite{AaronsonBrod}.  

The rest of this paper is organized as follows: \Cref{sec:GBS} provides a brief introduction to GBS as well the BipartiteGBS protocol for programming GBS devices according to arbitrary transition matrices. It then states important properties of BipartiteGBS with Gaussian transition matrices: two lemmas concerning photon-number statistics and normalization constants (proofs in \Cref{app:bounds}), and analytical formulas for the distribution of click statistics backed up by numerics (proofs in \Cref{app:clicks}). \Cref{sec:hardness_proof} contains the proof that classical simulation of BipartiteGBS in the no-collision regime is hard (\Cref{thm:intro_main}).  \Cref{sec:collision_reduction} deals with BipartiteGBS in the collision regime and contains proofs of theorems~\ref{thm:reduction_intro} and~\ref{thm:reduction_constant_collisions} (though proofs of some important lemmas in \Cref{app:lemma_proofs}).  Numerical evidence for extending our arguments beyond the dilute limit are given in \Cref{sec:beyond_dilute}.


\section{BipartiteGBS: Gaussian boson sampling with arbitrary transition matrices}\label{sec:GBS}

\subsection{Gaussian boson sampling introduction} \label{subsec:model}

\begin{figure}
\centering
    \includegraphics[width=0.75\columnwidth]{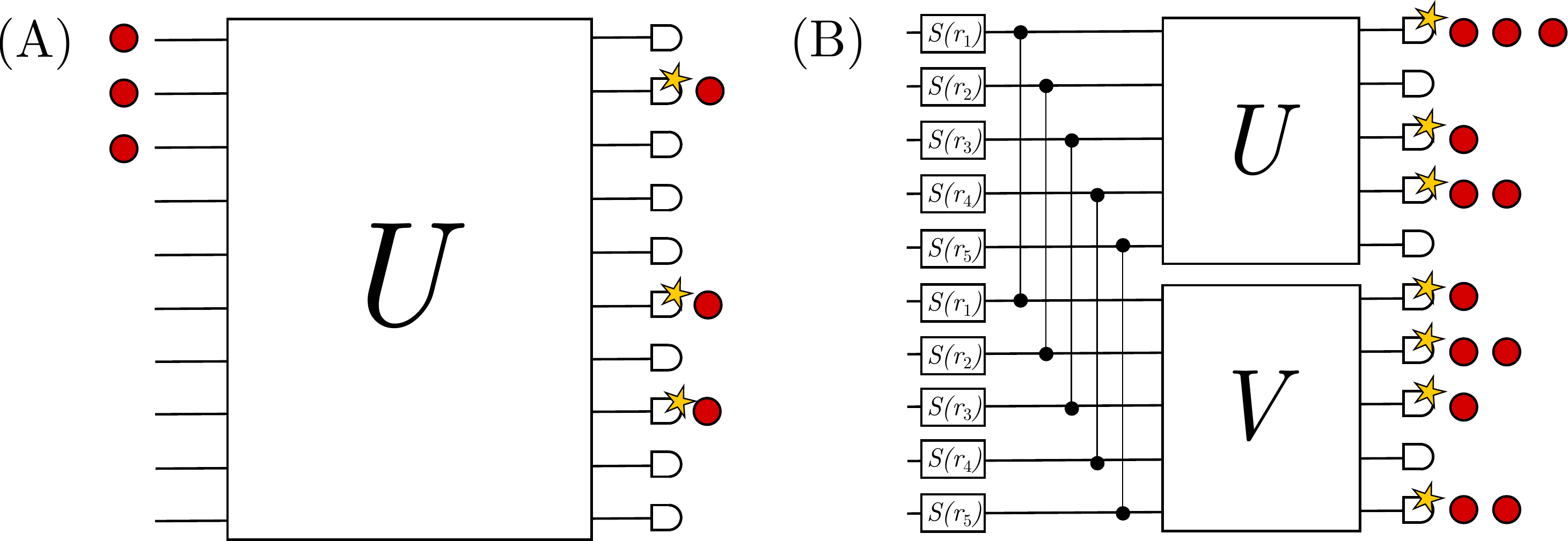}
    \caption{(A) Standard BosonSampling. The input to the interferometer are single photon states, which are detected at the output. (B) BipartiteGBS setup to encode arbitrary matrices. For a matrix $C$ with singular value decomposition $C = U \text{diag}(\sigma_i)V^T$, two-mode squeezed states with squeezing parameter $r_i = \tanh^{-1}(\sigma_i)$ are generated. These are sent to independent interferometers applying the unitaries $U$ and $V$, respectively.}
    \label{Fig: GBS}
\end{figure}

Gaussian boson sampling is a model of photonic quantum computation where a multi-mode Gaussian state is prepared and then measured in the photon-number basis~\cite{hamilton2017}. Gaussian states receive this name because their Wigner function---a quasi-probability representation of quantum states of light---is a Gaussian distribution~\cite{gqi2012}. We consider pure Gaussian states without displacements, which can be prepared from the vacuum by a sequence of single-mode squeezing gates followed by linear interferometry.

In contrast to BosonSampling, which uses single photons, GBS employs squeezed states as the input to the linear interferometer. In terms of the creation and annihilation operators $a_i$ and $a_i^\dagger$ on mode $i$, a squeezing gate is given by $S(r_i)=\exp[r_i(a_i^{\dagger 2}-a_i^2)/2]$, where $r_i$ is a squeezing parameter. A squeezed state can be prepared by applying a squeezing gate to the vacuum. A linear interferometer transforms the operators as
\beq
a_i \rightarrow a_i' = \sum_{j} U_{i j} a_j,
\eeq
where $U$ is a unitary matrix.

Let $S=(s_1, s_2, \ldots, s_{2m})$ encode a measurement outcome on $2m$ modes, where $s_i$ is the number of photons in mode $i$. The probability of observing sample $S$ when measuring a pure Gaussian state in the photon-number basis is given by \cite{hamilton2017}:
\beq \label{Eq: GBS_dbn}
\Pr(S) = \frac{1}{\norm}\frac{\left|\Haf(A_S)\right|^2}{s_1!s_2!\cdots s_{2m}!}.
\eeq
Here
\begin{align}
    \norm &= \prod_{i=1}^{2m} \cosh(r_i),\\
    A &= U\,\text{diag}(\sigma_1, \sigma_2, \ldots, \sigma_{2m})U^T,\label{Eq: A}\\
    \sigma_i &= \tanh(r_i), \hspace{0.3cm} 0\leq \sigma_i <1
\end{align}
where $r_i$ is the input squeezing parameter in the $i$th mode, $U$ is the unitary describing the interferometer, and $\Haf(\cdot)$ is the hafnian.\footnote{For a $2n \times 2n$ symmetric matrix $A$, $\Haf(A) = \frac{1}{n! 2^n} \sum_{\sigma \in \mathrm{S}_{2n}} \prod_{i=1}^n A_{\sigma(2i-1),\sigma(2i)}$.  See references~\cite{caianiello1953quantum, barvinok2016combinatorics, bjorklund2018faster} for more detailed discussion of the hafnian and its complexity.} The matrix $A_S$ is constructed from $A$ by repeating the $i$th row and column of $A$ $s_i$-many times (e.g., if $s_i = 0$, both the corresponding row and column are removed entirely). Notice that \Cref{Eq: A} implies that $A$ is symmetric.

The total mean photon number in the distribution is given by~\cite{jahangiri2020point}
\beq \label{Eq: mean_n}
 \sum_{i=1}^{2m}\frac{\sigma_i^2}{1-\sigma_i^2}.
\eeq
For future convenience we introduce $n$ as the random variable describing the total number of \emph{photon pairs} in a given sample. Its quantum-mechanical expectation is simply given by
\beq
\braket{n} = \frac12 \sum_{i=1}^{2m}\frac{\sigma_i^2}{1-\sigma_i^2}.
\eeq

Notice that it is possible to choose a parameter $\lambda>0$ and perform a rescaling $A_{\lambda} := \lambda A$ so that the singular values of $A_\lambda$ become $\lambda \sigma_i$.  One can check that the mean number of photon pairs $\braket{n_\lambda}$ is continuous for $\lambda \in [0, 1/\sigma_{\mathrm{max}})$ and grows arbitrary large, and so it is possible to set $\lambda$ such that $\braket{n_\lambda}$ is any desired non-negative number.

\subsection{BipartiteGBS}
We now introduce BipartiteGBS, a specific strategy for programming a GBS device such that the resulting distribution depends on an arbitrary transition matrix, not just a symmetric one. This scheme is illustrated in \Cref{Fig: GBS}. First, we construct a device with $2m$ modes and generate photons by applying two-mode squeezing gates to modes $i$ and $i+m$ for $i=1,2,\ldots,m$. The two-mode squeezing gate is defined as $S_2(r_i)=\exp[r_i(a_i^{\dagger}a_{i+m}^{\dagger}-a_i a_{i+m})]$. It can be decomposed as two single-mode squeezing gates with identical parameters (i.e., $r_i = r_{i+m}$) followed by a 50:50 beamsplitter. The subsequent interferometer is configured by applying a unitary $U$ to the first $m$ modes and a separate unitary $V$ to the second half of the modes. In this sense, this construction is a generalization of Scattershot BosonSampling~\cite{lund2014} and Twofold Scattershot BosonSampling~\cite{chakhmakhchyan2017boson}. In the former the second interferometer is fixed to $V=\id$ and $r_i = r$ for all the two-mode squeezers, while in the latter only the squeezing parameters are fixed.

In this setup, the GBS distribution is also given by \Cref{Eq: GBS_dbn}, but in this case the $A$ matrix satisfies
\begin{align}
    A &= \begin{pmatrix}
    0 & C \\
    C^T & 0
\end{pmatrix},\label{Eq:A=A(C)}    \\
C &= U \, \text{diag}(\tanh r_i ) V^T.
\end{align}

The expression for $C$ is equivalent to the singular value decomposition $C=U \Sigma W^\dagger$ of an arbitrary complex matrix with singular values $\sigma_i=\tanh(r_i)\in[0,1)$, where we simply set $V^T=W^\dagger$. Thus, it is possible to choose $C$ to be an arbitrary complex matrix, up to a rescaling $C\rightarrow \lambda C$ for some appropriate $\lambda>0$ such that the singular values satisfy $\sigma_i\in[0,1)$.

\begin{figure}[t!]
\centering
    \includegraphics[width=0.75\columnwidth]{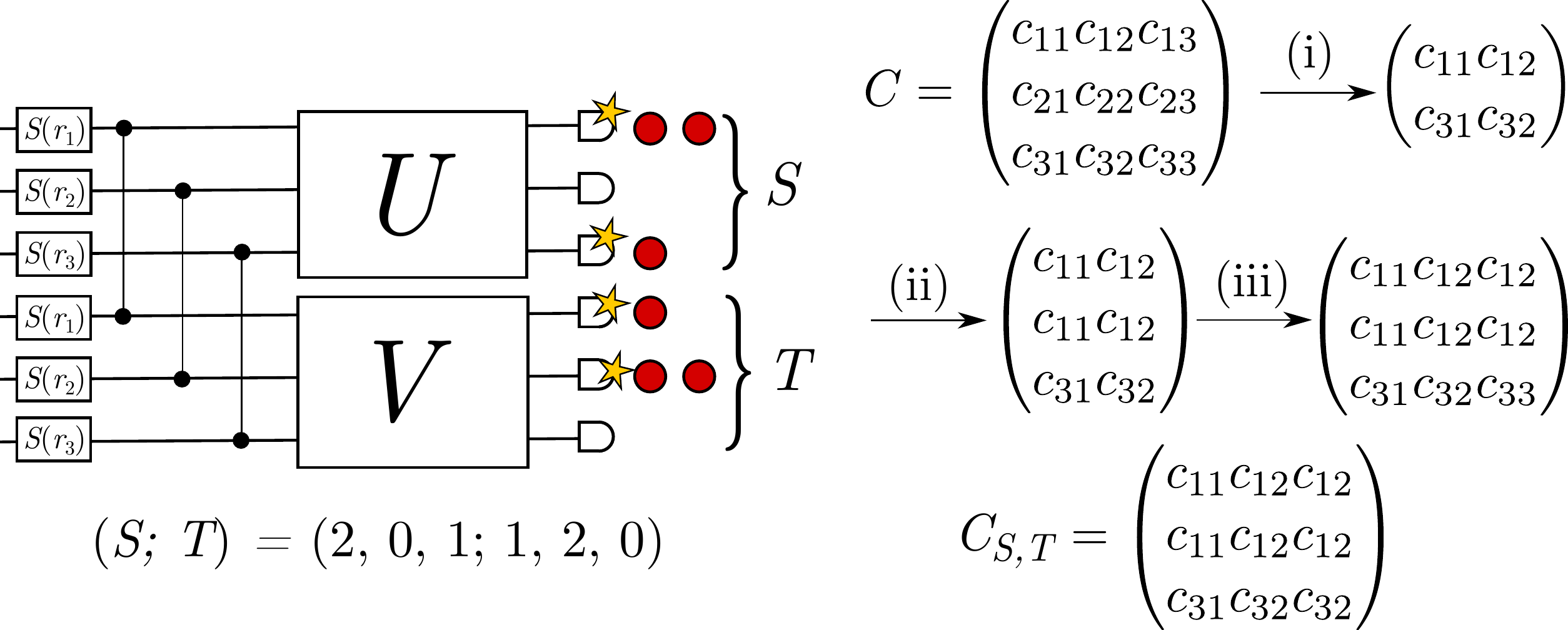}
    \caption{Constructing submatrices $C_{S,T}$ of the GBS distribution in \Cref{Eq: main_dbn}. In this example, a device with $2m=6$ modes is configured according to an $m$-dimensional matrix $C= U \, \text{diag}(\tanh r_i ) V^T$. The output is the photon pattern $(S;T)=(2,0,1;1,2,0)$, where $S$ determines the rows of the submatrix and $T$ determines the columns. In step (i), since $s_2=0$ and $t_3=0$, we remove the second row and the third column. In step (ii), because $s_1=2$, we repeat the first row twice and finally in step (iii) since $t_2=2$, we repeat the second column twice to arrive at the submatrix $C_{S,T}$. }\label{fig:C_ST_example}
\end{figure}

The resulting output distribution can be more elegantly expressed directly in terms of the matrix $C$, which we refer to as the \emph{transition matrix}. We introduce the notation $(S;T)=(s_1,\ldots,s_m;t_1,\ldots, t_m)$ to denote a sample. Here $s_i$ is the number of photons in mode $i$ and $t_i$ is the number of photons in mode $i+m$, for $i=1,2,\ldots,m$. From \Cref{Eq:A=A(C)}, the outcomes $S$ determine the rows of $C$ and the columns of $C^T$ that are kept or repeated when defining the submatrix $A_S$. Similarly, the outcomes $T$ determine the columns of $C$ and the rows of $C^T$. With this in mind, we employ the identity
\begin{align*}
\Haf \left[ \begin{pmatrix}
    0 & C \\
    C^T & 0
\end{pmatrix} \right] = \Per(C),
\end{align*}
to express the GBS distribution as:
\beq \label{Eq: main_dbn}
\Pr(S;T) = \frac{1}{\norm} \frac{\left\lvert \Per(C_{S,T})\right\lvert ^2}{\prod_i s_i! \prod_j t_j!}.
\eeq
The notation $C_{S,T}$ corresponds to a submatrix obtained as follows: if $s_i=0$, the $i$th row of $C$ is removed. If $s_i>0$, it is instead repeated $s_i$ times. Similarly, if $t_i=0$, the $i$th column of $C$ is removed and if $t_i>0$, it is repeated $t_i$ times.  See \Cref{fig:C_ST_example} for an example.

Since the permanent is only defined for square matrices, the number of photons detected in the first half of the modes should be equal to the total number of photons detected in the second set of modes, i.e., $\sum_{i=1}^m s_i = \sum_{i=1}^m t_i = n$. Physically, this corresponds to the action of the two-mode squeezing gate, which generates pairs of photons such that every photon in the first $m$ modes has a twin photon in the remaining $m$ modes. This observation, or the fact that $r_i = r_{i+m}$, allows us to write the expectation of the number of pairs as $\braket{n} = \sum_{i=1}^{m}\frac{\sigma_i^2}{1-\sigma_i^2}$, where the sum only extends to $m$.

\Cref{Eq: main_dbn} is the starting point for the computational problem we consider. In this formulation, our GBS construction is almost identical to a standard BosonSampling setup~\cite{aaronson2013}. In both cases, probabilities are given in terms of the permanents of submatrices constructed in the same manner. The main difference is that, in our case, we employ $2m$ modes to encode an arbitrary $m\times m$ complex matrix, whereas the corresponding matrix in BosonSampling must be unitary, namely equal to the matrix that describes the interferometer. Another crucial difference is the normalization factor $\norm$. It is necessary to account for the fact that the space of outcomes includes events with different total photon numbers, and it will influence the behaviour of errors in our final result in \Cref{sec:hardness_proof}.

\subsection{BipartiteGBS with Gaussian matrices} \label{subsec:gaussian}

As discussed above, it is possible to encode an arbitrary matrix in the GBS output distribution. In this section, we specialize this to the case of Gaussian random matrices. Let $\mathcal N(\mu, \Sigma^2)^{m \times m}_\mathbb C$ be the distribution over $m \times m$ matrices whose entries are independent complex Gaussians variables with mean $\mu$ and variance $\Sigma^2$. We choose $C \sim \mathcal N(0, \Sigma^2)^{m \times m}_\mathbb C$, for $\Sigma$ to be specified shortly.

By choosing $C$ to be Gaussian, this mirrors the case of BosonSampling \cite{aaronson2013}, where sufficiently small submatrices of uniformly Haar random unitaries are approximately also Gaussian. Therefore, this will allow us to support our hardness-of-simulation result on the same set of conjectures. However, in our construction, any $n \times n$ submatrix of $C$ is also Gaussian, for any scaling between $m$ and $n$. Contrast this with BosonSampling, where requiring submatrices to be approximately Gaussian formally constrains the number of modes to be much larger than the number of photons---rigorously, $m = \omega(n^5)$. We now prove a few important facts about GBS with Gaussian matrices.

First, several important quantities, such as the squeezing parameters, the normalization constant $\norm$, and the mean photon number $\braket{n}$, depend on the list of singular values of $C$, which we denote by $\{\sigma_i\}_{i=1\ldots m}$. For $m$-dimensional random complex matrices with mean $0$ and variance $1/m$, in the asymptotic limit $m \rightarrow \infty$, the distribution $p(\sigma)$ of singular values converges to
\beq
p(\sigma) = \frac{1}{\pi} \sqrt{4-\sigma^2}.
\eeq
This is the quarter-circle law for random matrices~\cite{shen2001singular}. Importantly, this result states that singular values are constrained within a finite interval, in this case $\sigma\in[0,2]$, with high probability. The probability that the largest singular value is greater than $2+\epsilon$ decays exponentially as $m \exp(-m \epsilon^2/8)$ \cite{haagerup_thorbjrnsen:2003}. See also \Cref{lem:max_eigen_wishart} in \Cref{subapp:boundN2}. For now, we simply assume that the singular values are in fact within this range, though we return to this issue in the next section.  The above equation holds as the limit for the empirical distribution over singular values of $C$, and therefore it also corresponds to the limit marginal distribution satisfied by any single $\sigma_i$. However, we cannot assume in general that the $\sigma_i$ are drawn \emph{independently} from it.

Recall from \Cref{Eq: mean_n} that for $C$ to be amenable to encoding in the GBS device, its largest eigenvalue must lie in the range $[0,1)$. Furthermore, it will be useful to tune the relation between $m$ and $\braket{n}$. To address both of these issues, we rescale the matrix by a further factor of $1/\alpha$, for some $\alpha > 2$. Alternatively, we can choose $C$ from $\G:=\mathcal N(0, 1/(\alpha^2 m))^{m \times m}_\mathbb C$. By doing this, the limiting distribution for the singular values in the interval $[0,2/\alpha]$ is
\beq \label{eq:quartercircle}
p_\alpha(\sigma) = \frac{\alpha}{\pi} \sqrt{4-\alpha^2 \sigma^2}.
\eeq
We now consider the typical behaviour of two quantities that will be important for our main result. The first is the number of photon pairs $n$. The size of the matrix permanent associated with an output probability---and hence its complexity---is directly determined by the number of photons observed in a given experimental run. However, in GBS the total photon number is not fixed, and the fluctuations in photon number depend on the matrix $C$ via its singular values. Therefore, it will be important to prove that fluctuations around the mean photon number are small enough so that our main argument is stable.

The other important quantity is the normalization constant $\norm$, which appears as a multiplicative factor between the matrix permanent and the output probabilities. For this reason, it will directly affect the error bounds of our main results, and we need to prove that it is typically not too large.

In the remainder of this section we give an intuitive analysis of these quantities, together with suitably formal bounds that are proven in \Cref{app:bounds}.

\subsubsection{Fluctuations of the total photon number} \label{subsubsec:photonumber}

From Eq.~\eqref{eq:quartercircle}, we can compute the expected mean number of photon pairs as we vary\footnote{Because it can sometimes be confusing, let us reiterate here that we use $\braket{\cdot}$ to denote the ``quantum-mechanical average'', i.e.\ average over the photon number distribution for given transition matrix, and $\mathbb E[\cdot]$ to denote the expected value over the Gaussian ensemble.  We use $\Delta^2[\cdot]$ and $\Var[\cdot]$ for the ``quantum-mechanical variance'' and variance due to the Gaussian ensemble, respectively.} $C$ over the Gaussian ensemble:
\beq \label{mean_scaling}
\mathbb E[\braket{n}] =  \sum_{i=1}^m \mathbb E\left[\frac{\sigma_i^2}{1-\sigma_i^2} \right] = \frac{m}{2} \left(\alpha^2\left(1-\sqrt{1-4/\alpha^2}\right)-2\right).
\eeq
Throughout our main argument, we usually consider a regime where $m$ scales faster than $\Omega(n)$, i.e., as $\Theta(n^2)$. This corresponds to a regime where $\alpha$ is large  and we can write $\alpha^2(1-\sqrt{1-4/\alpha^2})-2\approx\alpha^2(1-1+2/\alpha^2+2/\alpha^4)-2=2/\alpha^2$. Therefore in this regime we have
\beq
\mathbb E[\braket{n}] \approx  m/\alpha^2.
\eeq
For instance, if we want the GBS device to operate in a regime where $m = c \mathbb E[\braket{n}]^2$ for some $c>0$, it suffices to choose $\alpha=(cm)^{1/4}= \Theta(m^{1/4})$. We will refer to this regime from now on as the \emph{dilute limit}.

Even assuming the singular values follow the quarter-circle distribution exactly, computing the expectation of $\braket{n}$ is not enough. The complexity implied in Eq.~\eqref{Eq: main_dbn} depends on the observed number of photons, not on $\braket{n}$. Therefore we must prove that, with high probability, $n$ is not so far from its expectation as to invalidate our conclusions. We show the following formal bound, which follows from \Cref{thm:<n>_variance} and \Cref{thm:n_variance} in \Cref{app:bounds}:
\begin{lemma}\label{lem:boundn}
For any $\delta > 0$, we have
\begin{align*}
&\Pr\left[\left|\langle n \rangle - \frac{m}{\alpha^2}\right| \ge \frac{512 m}{\alpha^4} + \frac{1}{\alpha^2}\sqrt{\frac{2}{\delta}}\right] \le \delta
\\
&\Pr\left[ \left|n - \frac{m}{\alpha^2}\right| \ge  \frac{2 \sqrt{m}}{\alpha \sqrt{\delta}} +  \frac{3}{\alpha \delta^{3/4}} + \frac{84 \sqrt{m}}{\alpha^2\sqrt{\delta}} + \frac{512 m}{\alpha^4} \right] \le \delta
\end{align*}
whenever $\alpha \ge 6$, and $m \ge \ln(1/\delta)$. The first probability is over the choice of Gaussian matrix $C$, whereas the second probability is over both the choice of $C$ and over the photon number distribution.
\end{lemma}

Notice that in the dilute limit this statement implies that with high probability over the choices of Gaussian matrix and over the photon number distribution, the observed photon number is linear in $m/\alpha^2$ to leading order.

\subsubsection{The normalization factor} \label{subsubsec:normalization}

Let us now consider the typical behaviour of the normalization factor $\norm$. Recall that we can write
\beq
    \norm = \prod_{i=1}^{m} \cosh(\arctanh \sigma_i)^2 = \prod_{i=1}^{m} \frac{1}{1-\sigma_i^2}.
\eeq
Assuming $\sigma_i \in [0, 2/\alpha]$, it holds that $\norm \in [1,1/(1-4/\alpha^2)^m]$. Asymptotically, the upper bound can be written as
\beq \label{eq:norm_max}
\norm_\textrm{max} = 1/(1-4/\alpha^2)^m  \approx e^{4m/\alpha^2}.
\eeq
We want $\norm$ to be as small as possible. As will become clear in the next sections, the bound in \Cref{eq:norm_max} is not sufficiently tight for our purposes. On the other hand, if each singular value $\sigma_i$ is drawn independently from the quarter-circle distribution (\Cref{eq:quartercircle}), then the expectation of $\norm$ would scale more favourably as $e^{m/\alpha^2}$. Although the singular values are not independent, we prove that this heuristic argument in fact provides the right scaling (see \Cref{subapp:boundN2}). More specifically, we give the following bound:
\begin{lemma} \label{lem:boundZ}
For any $\delta > 0$
\beq
\Pr \left[\norm \ge  \frac{2}{\delta} e^{m/\alpha^2} e^{272 m / \alpha^4} \right] \le \delta,
\eeq
whenever $\alpha \ge 6$, and $m \ge \ln(1/\delta)$.
\end{lemma}    

Recall that $\alpha=\Theta(m^{1/4})$ in the dilute limit, so \Cref{lem:boundZ} implies that $\norm$ is bounded asymptotically by $e^{m/\alpha^2}$.

\subsubsection{Collision statistics beyond the dilute limit}\label{sec:collisions}
A BipartiteGBS sample is specified by $M=2m$ non-negative integers giving the number of photons measured in each of the modes:
\eq{
    (S,T) = (s_1,\ldots,s_m;t_1,\ldots,t_m)   .
}
As mentioned before, the total number of photons detected in the first $m$ modes should be equal to the total number of detected in the second $m$ modes---that is, $n=\sum_{i=1}^m s_i = \sum_{j=1}^m t_j$.

Another useful variable to consider is the number of clicks. A click sample is obtained from a photon number sample by ``thresholding'' the events, mapping any event with more than zero photons into outcome 1 while mapping vacuum events to 0. We write these thresholded samples as
\beq
(D,E) =(d_1,\ldots,d_m;e_1,\ldots,e_m)
\eeq
where $d_i,e_j = 1$ if $s_i,t_j \ge 1$ and $0$ otherwise.
It is also useful to write the total number of clicks in either half
\eq{
    d &= \sum_{i=1}^m d_i \leq n, \quad e = \sum_{j=1}^m e_j \leq n,
}
where we also state the obvious fact that the total number of clicks is always smaller or equal to the total number of photons detected. Note that, \emph{unlike} the photon number, the number of clicks in both  halves of the modes need not be the same, thus in general, $d \neq e$. Whenever the number of clicks is less than the number of photons, there must be collisions (at least one mode with more than one photon).

To understand why the number of clicks is an important random variable consider the expression for the probabilities  (recall~Eq.~\eqref{Eq: main_dbn}) depending on
\eq{
    \text{Per}(C_{S,T}).
}
A priori, while the $n \times n$ matrix $C_{S,T}$ may be large, its rank ($\min(d,e)$) may be small.  Because matrices of small rank have efficient algorithms \cite{barvinok1996two}, it is useful to understand the statistics of the clicks in each of the two halves of the modes. In the dilute limit no-collision events dominate and thus $d=e=n$.

\begin{figure}
\centering
    \includegraphics[width=0.31\textwidth]{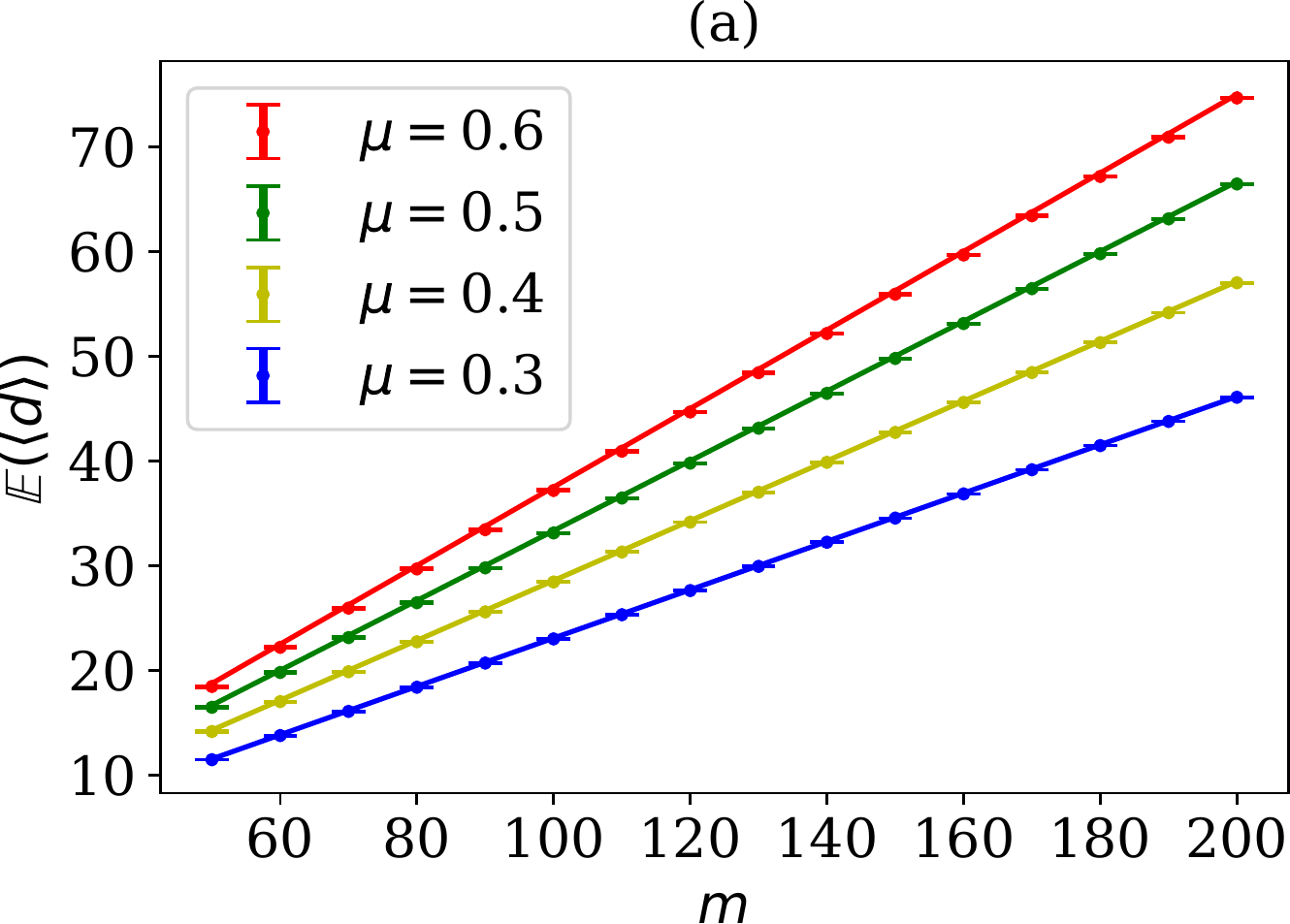}
    \includegraphics[width=0.31\textwidth]{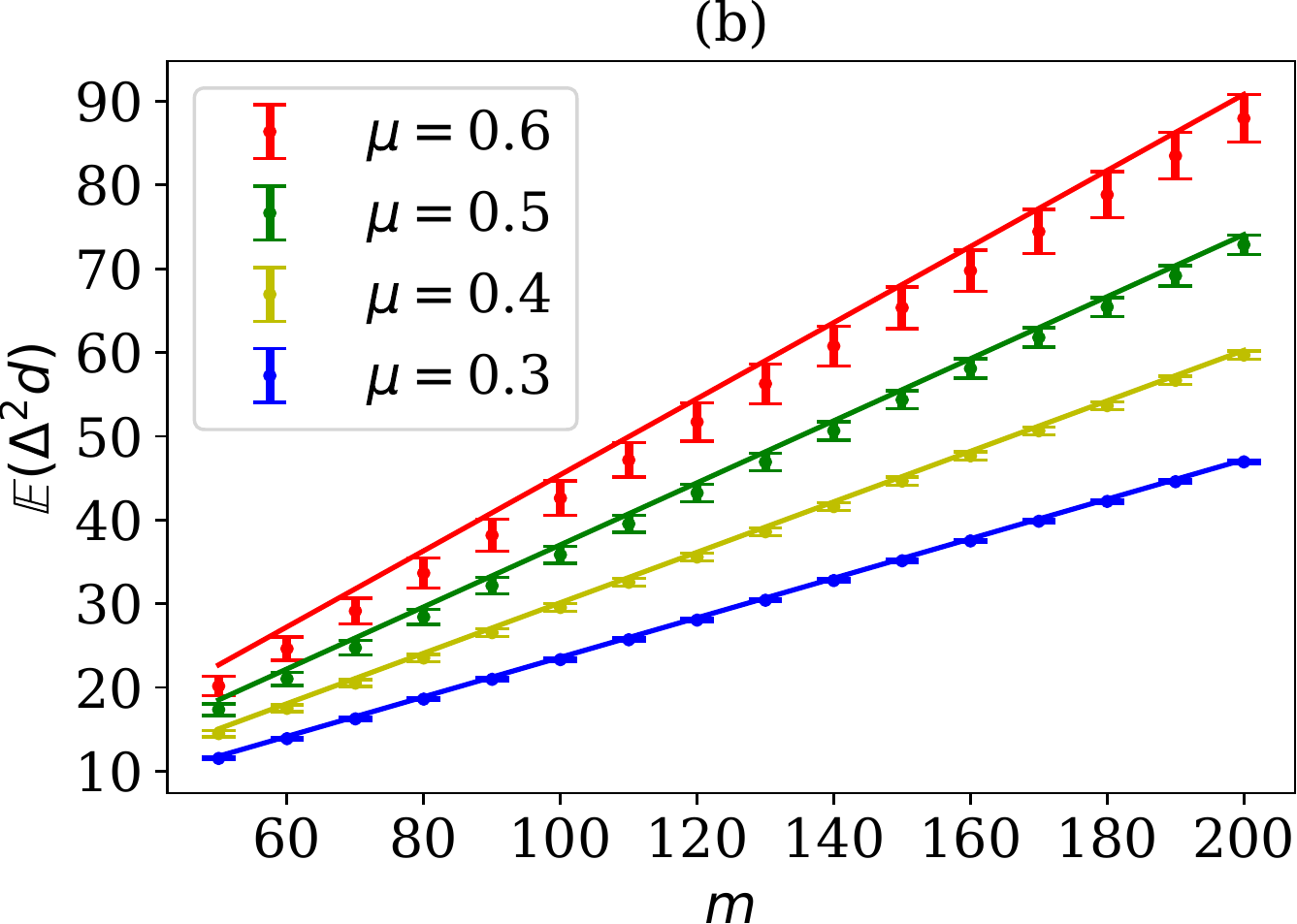}
    \includegraphics[width=0.32\textwidth]{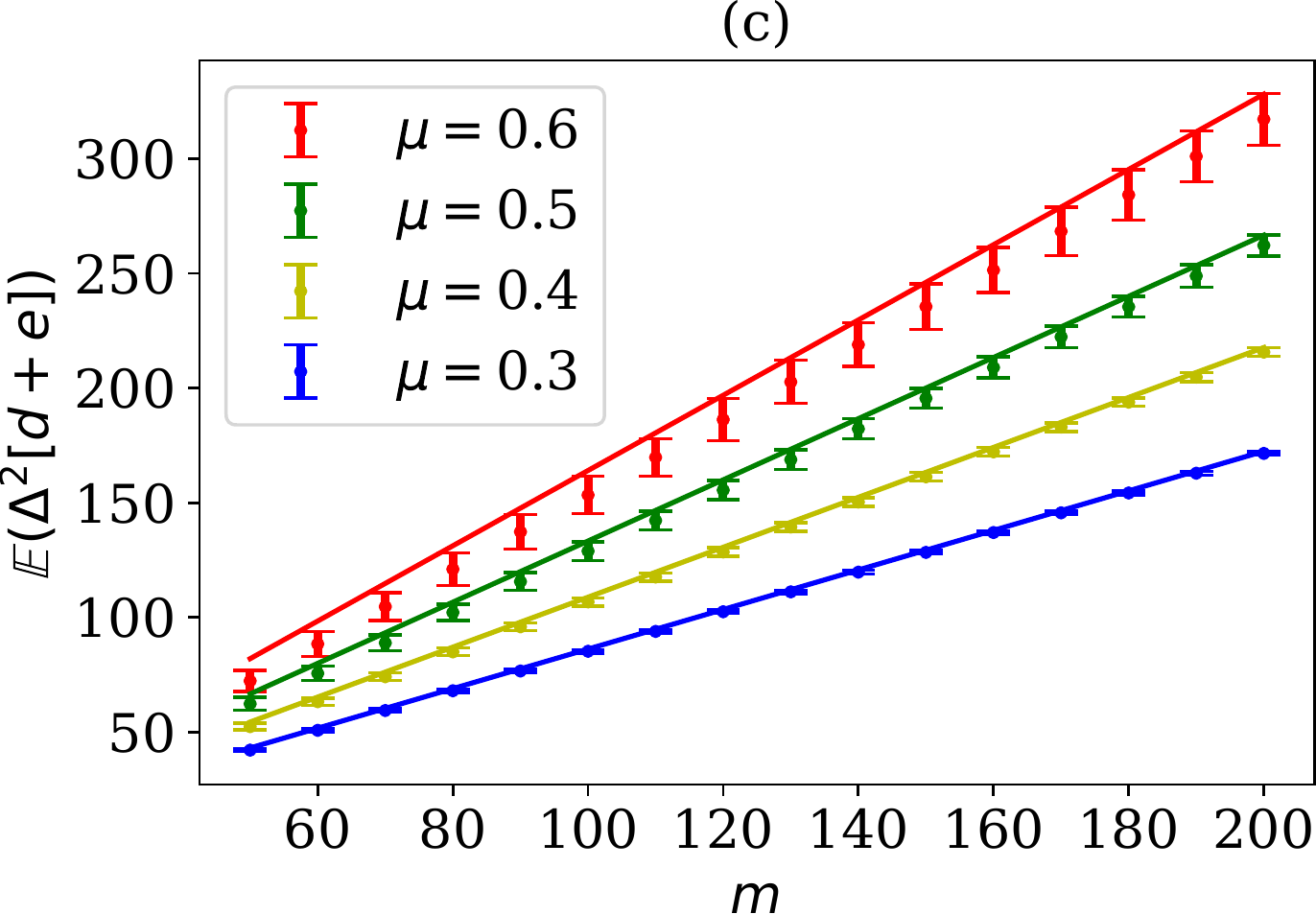}
    \caption{\label{fig:comparison}We compare the analytical results in Eqs.~\eqref{app:means}, \eqref{equations} with numerics. For each value of $m$ we generate 2500 random matrices where each of the $m^2$ elements is drawn i.i.d. from the standard complex-normal distribution and then each matrix is scaled so as to fix the total mean photon number to be $2\braket{n} = \mu \times 2m$.
        We then calculate (a) the average of the mean number of clicks in the first half of the modes $\mathbb{E}(\braket{d})$, (b) its variance $\mathbb{E}\left[ \Delta^2 d \right]$ and (c) the variance in the total number of clicks $\mathbb{E}\left[ \Delta^2(d+e) \right]$. The error bars in (a), (b) and (c) are given by  $\sqrt{\text{Var}(\braket{d})}$, $\sqrt{\text{Var}\left[ \Delta^2 d \right]}$ and $\sqrt{\text{Var}\left[ \Delta^2 [d+e] \right]}$ respectively and are obtained by using the 2500 Monte Carlo samples for each value of $\mu$ and $m$. The lines are the predictions from theory. We do not show the mean $\mathbb{E}(\braket{e})$ and variance $\mathbb{E}\left[ \Delta^2 e \right]$ of the clicks in the second half of the modes since they are indistinguishable from the corresponding results of the first half.
        Note that for densities $\mu < 0.5$ the theory agrees very well with the numerics.
        For larger values of the density the first-order Taylor expansions used to derive Eqs.~\eqref{equations} in Appendix~\ref{sec:coup} no longer hold, giving a significant deviation from the numerically obtained values as seen above for the variances in Figs. 3 (b) and 3 (c).
        The numerical calculation was performed using The Walrus~\cite{walrus}.
    }
\end{figure}

Beyond the dilute limit, we can find simple expressions for the first and second moments of the total number of clicks in either half of the $2m$ modes. The detailed derivation of these results is provided in \Cref{app:clicks}. These expressions are written in terms of the photon number density
\eq{\label{eq:density}
    \mu =  \mathbb{E}( \braket{a_i^\dagger a_i}) = \frac{\mathbb{E}(\braket{n})}{m}.
}
For the first order moments (means) we simply invoke the fact that the
 reduced states of two-mode squeezed states are thermal states, and scrambling from the random interferometers
leads also to locally thermal states to find,
\eq{
\mathbb{E}(\braket{d}) =\mathbb{E}(\braket{e})&=  m \frac{\mu}{1+\mu}\label{app:means}.
}
By rewriting the density $\mu$ in terms of the mean number of pairs and the number of modes one easily derives Eq.~\eqref{eq:harmonic_mean} for $\mathbb{E}(\braket{d} + \braket{e})$.
For the second order moments we need to invoke the quarter circle law and use a Taylor expansion to obtain
\seq{\label{equations}
    \eq{
        \mathbb{E}(\Delta^2 d) =\mathbb{E}(\Delta^2 e) &= m \frac{\mu  \left(1-\mu ^2+\mu\right) }{(1-\mu ) (\mu +1)^3}  = \left(1-\mu ^2+\mu\right) \mathbb{E}(\text{Cov}(d,e)) \label{app:single_cov},\\
        \mathbb{E} (\Delta^2(d+e)) &= 2m \frac{(2-\mu) \mu}{(1-\mu) (1+\mu)^2}.
    }
}
In the dilute limit $\mu \sim 1/\sqrt{m} \ll 1$ we have  $\mathbb{E}(\Delta^2 d) =\mathbb{E}(\Delta^2 e) \approx (1+\mu) \mathbb{E}(\text{Cov}(d,e))$ which tells us that $d$ and $e$ are very strongly correlated, as one would expect since in this limit clicks reduce to photon numbers, which should be equal in the two halves of the modes. Even beyond the dilute limit, the equations above predict strong correlations between the number of clicks in either half of the modes. We can for example consider the linear correlation ratio
\begin{align}
\text{Corr}(e,d) = \frac{\mathbb{E}(\text{Cov}(d,e)) }{\sqrt{\mathbb{E}(\Delta^2(d)) \mathbb{E}(\Delta^2(e))}} = \frac{1}{1+\mu - \mu^2}
\end{align}
which for a non-negligible density of $\mu=0.3$ gives $\text{Corr}(e,d) \approx 0.83$.
We find excellent agreement between the results in the equations above against exact numerical calculations for varying photon-number densities as the number of modes increased in \Cref{fig:comparison}.

We conclude this section with a discussion of the bosonic birthday paradox, a bound on how likely we are to observe collision events in a BosonSampling experiment, which will also turn out to govern BipartiteGBS~\cite{Arkhipov2012}. Specifically, for a BosonSampling experiment with $m$ modes and $n \sim c\, m^{(j-1)/j}$ photons, the bosonic birthday paradox says that the number of output modes where we expect to observe exactly $j$ photons converges to a Poisson random variable with mean $c^j$. Collision outcomes with more than $j$ photons are suppressed.

The key idea in the proof of the bosonic birthday paradox is that, upon applying a Haar-random matrix to an $n$-photon Fock state, the output state is the maximally mixed state over the $n$-photon, $m$-mode Hilbert space. While BipartiteGBS experiments do not in general have Fock input states (in fact, the total photon number is a random variable), one can show that the bosonic birthday paradox holds when we restrict to a fixed subspace of $n$ photons since that restricted input state \emph{is} a Fock state. Moreover, in the singular value decomposition of a Gaussian matrix, the two unitary matrices $U$ and $V$ corresponding to the interferometers in BipartiteGBS are Haar-random~\cite{tulino2004random}. Therefore, it will still be the case that the output state, when averaged over the Gaussian ensemble, can be seen as a copy of the maximally mixed state on each set of modes, as per \Cref{Fig: GBS}.  

For example, because \Cref{lem:boundn} shows that $n$ is highly concentrated around $m/\alpha^2$ in the dilute limit, we can apply the bosonic birthday paradox in this regime (at $j=2$).  In particular, this implies that some constant fraction of the output distribution of a BipartiteGBS experiment has no collisions, a fact which is somewhat implicit\footnote{\Cref{thm:main} implies the following weaker statement: supposing the Permanent Anti-Concentration Conjecture, a polynomially large fraction of the BipartiteGBS output distribution has no collisions. \textit{Proof.}~Suppose that the no-collision subspace $\mathcal H$ does not comprise at least an inverse-polynomial fraction of the BipartiteGBS distribution---that is, $\sum_{S \in \mathcal H} p_S = o(1/\poly(m))$ with overwhelming probability over Gaussian matrices. Then, a randomly chosen probability $p_S$ in the no-collision subspace must be $o(1/(|\mathcal H| \poly(m)))$ with high probability. Expanding out the expression for the probability and rearranging, we get
$$
\frac{|\mathcal H|}{\mathcal Z \alpha^{2n} m^n} |\Per(C_S)|^2 = o(1/\poly(m))
$$
where $C \sim \mathcal N(0,1)_{\mathbb C}^{m \times m}$ and $S$ is chosen uniformly at random from $\mathcal H$. We show in \Cref{thm:main} (Eq.~\eqref{eq:I_bound} and Eq.~\eqref{eq:Z_bound}) that 
$\mathcal Z \alpha^{2n} m^n/|\mathcal H| n! = \poly(m)$
with overwhelming probability. Plugging in this bound, we get $|\Per(C_S)|^2 = o(n!/\poly(m)).$
However, this contradicts the Gaussian Permanent Anti-Concentration Conjecture which says that $|\Per(C_S)|^2 = \Omega(n!/\poly(m))$ with high probability.
}
in the proof of \Cref{thm:main}.


\section{Approximate average-case hardness for GBS in the dilute limit}\label{sec:hardness_proof}

In this section we prove our main result, namely that GBS is hard to simulate efficiently on a classical computer in the dilute limit and up to the same complexity conjectures as in BosonSampling~\cite{aaronson2013}. To that end, we start with the following computational problem:

\begin{problem}
\label{prob:GPEpa}
\emph{($\GPEpa$ \cite{aaronson2013})} Given as input a Gaussian matrix $X \sim \mathcal N (0,1)^{n \times n}_{\mathbb C}$ together with error bounds $\epsilon, \delta > 0$, output a complex number $\tilde P$ such that $\left|\tilde P - \left|\Per X\right|^2\right|\leq \epsilon n!$ with probability at least $1-\delta$.
\end{problem}

Our goal is to prove that, if there exists an efficient classical algorithm to simulate the output of a GBS experiment to high precision in total variation distance, then $\GPEpa$ can be solved in the complexity class $\FBPP^\NP$. It is conjectured\footnote{In \cite{aaronson2013}, the conjecture is decomposed into two parts:  the Permanent-of-Gaussians Conjecture asserts that the multiplicative version of the Gaussian permanent estimation problem is hard, while the Permanent Anti-Concentration Conjecture implies that this multiplicative version is poly-time equivalent to $\GPEpa$.}  that $\GPEpa$ is $\#\mathsf{P}$-hard, and so we obtain that the polynomial hierarchy ($\PH$) collapses to its third level by the usual chain of inclusions:
\beq
\P^\PH = \P^{\#\P} \subseteq \BPP^{\GPEpa} \subseteq \BPP^{\FBPP^\NP} \subseteq \Sigma_3^p.
\eeq
Since the polynomial hierarchy is widely conjectured to be infinite, such an efficient classical algorithm is unlikely to exist \cite{aaronson2013,bremner2010iqp}.

We will break the proof of our main result in two parts.  In \Cref{subsec:Stockemeyer}, we apply a Stockmeyer counting argument to leverage an efficient classical algorithm that samples from the GBS output distribution into a $\FBPP^\NP$ algorithm that produces an estimate of any probability within this distribution. This follows the corresponding argument in \cite{aaronson2013}, though we emphasize some aspects of the proof that are particular to our scenario. In \Cref{subsec:probtoperm}, we then use \Cref{Eq: main_dbn} to connect the output probability of an event to the permanent of a Gaussian matrix, and discuss how this affects the corresponding error bounds.

\subsection{From distributions to probabilities} \label{subsec:Stockemeyer}

Consider the probability distribution $\mathcal D_C$ at the output of a GBS experiment which implements an arbitrary\footnote{If the singular values of $C$ are greater than $1$, then we cannot implement the corresponding GBS experiment and $\mathcal D_C$ is undefined. For such matrices, let us just assume that $\mathcal D_C$ is the distribution that always outputs 0 photons in every mode.} transition matrix $C$, as described in \Cref{subsec:model}. Let $\mathcal O$ be a deterministic classical algorithm that takes as input a description of $C$, together with an error bound $\beta$ and a uniformly-random number $r$, and outputs a sample drawn according to distribution $\mathcal D_C'$. We write this as
\beq
\mathcal O(C, 0^{1/\beta}, r) \sim \mathcal D_C',
\eeq
as $r \in \{0,1\}^{\poly(m)}$ is sampled from the uniform distribution. Suppose also that
\beq
\left\lVert \mathcal D_C - \mathcal D_C' \right\rVert \le \beta.
\eeq
Our goal is to use Stockmeyer's theorem \cite{stockmeyer1983complexity} to show that a $\BPP^{\NP}$ machine, with access to $\mathcal{O}$, can produce a good approximation to the probability of some particular outcome.

At this point, for generality, we do not yet impose the restrictions that we are in the dilute limit, nor that $C$ is a Gaussian matrix. We only use the fact that the Gaussian ensemble is \emph{invariant} under permutations of its rows and columns. At a high level, this allows us to randomly permute the rows and columns of $C$ in order to hide the outcome $S$ that we care about within $\mathcal D_C$. The classical algorithm $\mathcal{O}$ is allowed to incur some total error $\beta$ in its probabilities, but it is unlikely to concentrate too much of that error in the specific hidden outcome.

Now let $S$ be a particular GBS outcome,\footnote{For convenience, we are writing $(S,T)$ as $S$ here. That is, $S$ is being used to represent both the row collisions \emph{and} column collisions.} and let $\mathcal H_S \subseteq \mathbb N^{2m}$ be the subset of GBS outcomes that contains (without multiplicity) all outcomes $S_{\pi,\sigma} = (s_{\pi(1)}, \ldots, s_{\pi(m)}; t_{\sigma(1)}, \ldots, t_{\sigma(m)})$ for all permutations $\pi, \sigma \in \mathrm S_m$. This set is, by definition, invariant under permutations within the first $m$ modes and within the last $m$ modes, which correspond to permutations of rows or columns of $C$, respectively. To reiterate, to work in the dilute limit, $S$ will be a $0/1$-vector, but such a restriction is not yet needed. Indeed, we could choose $S$ to be any possible GBS outcome, even those with many collisions. That said, in order to obtain an approximation to the underlying permanent question, the subspace $\mathcal H_S$ must correspond to outcomes with a sufficiently large probability mass.

We then have the following result, which closely follows Theorem 3 of \cite{aaronson2013}:

\begin{theorem}
\label{thm:additive_estimation}
There exists an $\FBPP^{\NP^{\mathcal O}}$ algorithm which, given an outcome $S$, produces an estimate of the corresponding probability $\mathcal D_C(S)$ to additive error $\varepsilon / \left|\mathcal H_S\right|$ with probability at least $1-\delta$ (over choices of $C$) in time $\mathsf{poly}(m,1/\varepsilon,1/\delta)$.
\end{theorem}
\begin{proof}
At a high level, we will construct an algorithm which counts all settings of the random bits that cause $\mathcal O$ to output $S$.  However, $\mathcal O$ is allowed to err with some constant probability $\beta$.  In order to prevent $\mathcal O$ from concentrating more of that error on $S$, we will first permute the rows and columns of $C$ uniformly at random.  Equivalently, we can say that $\mathcal O$ cannot tell which outcome in $\mathcal H_S$ we desire to approximate.

Set $\beta := \varepsilon\,\delta / 16$ and feed the input $\langle C, 0^{1/\beta}, r \rangle$ to $\mathcal O$.  Recall that $\mathcal O$ returns a sample from $\mathcal D_C'$ such that $\lVert \mathcal D_C - \mathcal D_C' \rVert \le \beta$ if $r$ is sampled uniformly.  Furthermore, if we let $p_X := \mathcal D_C(X)$ and $q_X := \mathcal D'_C(X)$ for any outcome $X \in \mathbb N^{2m}$, we have
\beq
q_S = \Pr_{r \in \{0,1\}^{\poly(m)}} \left[ \mathcal O(C, 0^{1/\beta}, r) = S \right].
\eeq
We use Stockmeyer's approximate counting method \cite{stockmeyer1983complexity} on this probability.  If we can also show that $q_S$ is close to $p_S$, this will imply that we can use Stockmeyer's method to estimate the desired quantity $p_S$ as well.

Let $\Delta_X := |p_X - q_X|$.  We get
\beq
\underset{X \in {\mathcal H_S}}{\mathbb E} [\Delta_X] \leq \frac{1}{|{\mathcal H_S}|} \sum_{X \in \mathbb N^{2m}} \Delta_X = \frac{2}{|{\mathcal H_S}|} \left\| \mathcal D_C - \mathcal D_C' \right\| \leq \frac{2 \beta}{|{\mathcal H_S}|}.
\eeq
By Markov's inequality, we get
\beq
\Pr_{X \in {\mathcal H_S}} \left[ \Delta_X > \frac{2 \beta k }{ |{\mathcal H_S}|} \right] < \frac{1}{k}.
\eeq
Setting $k := 4/\delta$ we have $2 \beta k = \varepsilon / 2$, and so
\beq
\Pr_{X \in {\mathcal H_S}} \left[ \Delta_X > \frac{\varepsilon}{2} \cdot \frac{1}{ |{\mathcal H_S}|} \right] < \frac{\delta}{4}.
\eeq
Our goal is to bound $\Delta_S$, the error in probability for the specific outcome we are trying to compute. Recall that rows and columns of $C$ were randomized before feeding it to $\mathcal O$. Since the distribution over $C$ is permutation invariant, we have that a random outcome of $C$ has the same probability as a fixed outcome which has been permuted randomly, that is,
\beq
\Pr_{(\pi, \sigma), C} \left[ \Delta_S > \frac{\varepsilon}{2} \cdot \frac{1}{ |{\mathcal H_S}|} \right] < \frac{\delta}{4},
\eeq
where $\pi, \sigma$ are the permutations used to randomize the rows and columns of $C$.\footnote{We note that this step takes place of the ``Hiding Lemma'' in \cite{aaronson2013}.}

Let us now formally describe the approximation we get for $q_S$ using Stockmeyer counting~\cite{stockmeyer1983complexity}. For any $\theta > 0$, we can obtain an estimate $\widetilde{q}_S$ such that
\beq
\Pr\left[ |\widetilde q_S - q_S| > \theta\, q_S \right] < \frac{\delta}{2}
\eeq
with an $\FBPP^{\NP^{\mathcal O}}$ machine running in time polynomial in $m$, $1/\theta$, and $1/\delta$.  Since $q_X$ is a probability, observe that $\mathbb E_{X \in {\mathcal H_S}} [q_X] \le 1/|{\mathcal H_S}|$ since $q_X \le 1$.  Once again we apply Markov's inequality to get
\beq
\Pr_{X \in {\mathcal H_S}} \left[ q_X > \frac{k}{|{\mathcal H_S}|} \right] < \frac{1}{k}.
\eeq
For similar reasons as before, it follows that
\beq
\Pr_{(\pi, \sigma), C} \left[ q_S > \frac{k}{|{\mathcal H_S}|} \right] < \frac{1}{k}.
\eeq
Finally, we set $\theta := \varepsilon\, \delta / 8$ and $k := 4/\delta$.  Combing all the above with the union bound, we get
\begin{align}
\Pr \left[ |\widetilde q_S - p_S| > \frac{\varepsilon}{|{\mathcal H_S}|} \right]
&\le \Pr \left[ |\widetilde q_S - q_S| > \frac{\varepsilon}{2} \cdot \frac{1}{|{\mathcal H_S}|} \right] +
\Pr \left[ |q_S - p_S| > \frac{\varepsilon}{2} \cdot \frac{1}{|{\mathcal H_S}|} \right] \nonumber\\
&\le \Pr\left[q_S > \frac{4}{\delta |{\mathcal H_S}|} \right] +
\Pr \left[ |\widetilde q_S - q_S| > \theta q_S \right] +
\Pr \left[ \Delta_S > \frac{\varepsilon}{2} \cdot \frac{1}{|{\mathcal H_S}|} \right] \nonumber\\
&<  \frac{\delta}{4} + \frac{\delta}{2} + \frac{\delta}{4},
\end{align}
where the probability is over choices of $C$, $(\pi,\sigma)$ and the randomness of the approximate counting procedure.
\end{proof}

\subsection{From probabilities to permanents} \label{subsec:probtoperm}

We just showed that, from the assumption that there exists a classical algorithm $\mathcal O$ that simulates the output distribution of a GBS experiment to within total variation distance $\beta$, it follows that there exists an $\FBPP^{\NP^\mathcal{O}}$ algorithm that produces an estimate to any outcome $S$ to within error $\varepsilon / |\mathcal{H}_S|$. We now show how to leverage this result to obtain an estimate for the permanent of a Gaussian matrix, i.e., to solve $\GPEpa$.

We begin by embedding the matrix we care about, $X$, within the GBS transition matrix $C$. Recall that $X \sim \mathcal N (0,1)^{n \times n}_{\mathbb C}$, and so there is an efficient algorithm which takes $X$ as input and produces a matrix $C' \sim \mathcal N (0,1)^{m \times m}_{\mathbb C}$ which contains $X$ as its submatrix, occurring in a random position.

Now recall that $C'$ cannot be directly implemented in the GBS setup, since we require its singular values to be between $[0,1)$. Furthermore, we wish to work in the dilute limit. Therefore, we rescale the matrix $C'$ by a factor of $1/(\alpha \sqrt{m})$, as discussed in \Cref{subsec:gaussian}, resulting in the transition matrix $C \sim \G$.

Suppose now, without loss of generality, that $X/(\alpha\sqrt{m})$ appears as the top left submatrix of $C$. Consider now the 2$n$-photon no-collision outcome, $S$, which contains a single photon in each of the modes $\{1, \ldots, n\}$ and of the modes $\{m+1, \ldots m+n\}$. By the discussion in \Cref{subsec:model}, the probability of this outcome is
\beq \label{eq:probtoperm}
\Pr(S)= \frac{1}{\norm}\frac{1}{m^n \alpha^{2n}} \left|\Per X\right|^2.
\eeq
From this correspondence between the probability of outcome $S$ and $\left|\Per X\right|^2$, together with \Cref{thm:additive_estimation}, we obtain the following result:
 
\begin{corollary}
\label{cor:perapprox}
 Let $X \sim \mathcal N (0,1)^{n \times n}_{\mathbb C}$. There exists an $\FBPP^{\NP^{\mathcal O}}$ algorithm which estimates $\left|\Per X\right|^2$ to additive error
 \beq
 \varepsilon \frac{\norm m^n \alpha^{2n}}{|{\mathcal H_S}|},
 \eeq
with probability at least $1-\delta$ over choice of $X$ in time  $\poly(m,1/\varepsilon,1/\delta)$.
\end{corollary}
\begin{proof}
We begin by following the procedure described previously in order to embed $X$ as a submatrix of the GBS transition matrix $C$. The entire procedure will fail if $C$ has singular values greater than 1 since such a transition matrix does not correspond to a valid GBS experiment.  However, by \Cref{lem:max_eigen_wishart}, the maximum singular value of $C$ is greater than $2/\alpha + \epsilon$ with probability at most $m e^{-m \alpha^4 \epsilon / 8}$.  Let us now set $\epsilon$ to $1 - 2/\alpha$.  If $m e^{-m \alpha^4 \epsilon / 8}$ is less than, say $\delta/2$, we can assume the singular values are less than $1$ by the union bound.  However, one can check that the experiment might fail with probability greater than $\delta/2$ if $1/\delta = \Omega(\exp(m \alpha^4))$.  In this case, however, our algorithm can run in time exponential in $m$ since every polynomial in $\exp(m)$ is bounded by another polynomial in $1/\delta$ (and recall we allow that our algorithm runs in time $\poly(1/\delta)$). Since $\left|\Per X\right|^2$ can be computed \emph{exactly} in time exponential in $m$ by Ryser's formula, the theorem still holds even in this extreme regime of $\delta$.

We now apply \Cref{thm:additive_estimation}, which provides an estimate of $\Pr(S)$ to within additive error $\varepsilon/|\mathcal{H}_S|$. By Eq.~\eqref{eq:probtoperm}, this immediately yields an additive estimate for $\left|\Per X\right|^2$ to within the desired error.
\end{proof}

Notice that we have yet to use the fact that we are working in the dilute limit---indeed, the previous theorems have been stated for a general scaling parameter $\alpha$.  We will now need to make this restriction. That is, we will show that any classical oracle $\mathcal O$ which approximately samples from the output distribution of our GBS experiment in the dilute limit can be leveraged to compute Gaussian permanents.

The main outstanding step in this proof is to find a sufficiently large subspace $\mathcal H_S$ in which to hide the outcome $S$, such that the error bound obtained in \Cref{cor:perapprox} matches the required error to solve $\GPEpa$. For that purpose, we restrict ourselves to the no-collision subspace with $n$ photon pairs, where $n$ is the size of the matrix permanent we wish to estimate in $\GPEpa$. More specifically, we set the scale parameter $\alpha$ such that the expected   number of photon pairs is exactly $n$, and restrict our attention to states $S = (s_1, \ldots s_m; t_1, \ldots t_m)$ where each $s_i$ or $t_j$ is only equal to 0 or 1, and where $\sum s_i = \sum t_j = n$. To clarify, the size of the matrix is fixed and \emph{not} a random number.  Nevertheless, we call this parameter $n$ suggestively since the total photon number of the experiment is unlikely to be too far from $\mathbb E [\braket{n}]$ by \Cref{lem:boundn}.  Remarkably, however, our proof never explicitly invokes this fact.

\begin{theorem}
\label{thm:main}
$\GPEpa \in \FBPP^{\NP^\mathcal{O}}$.
\end{theorem}
\begin{proof}
Recall that we are given a matrix $X \sim \mathcal N_{\mathbb C}^{n \times n}$, and we wish to construct an $\FBPP^{\NP^\mathcal{O}}$ algorithm which estimates $|\Per(X)|^2$ to additive error $\epsilon n!$ with probability at least $1-\delta$ in time $\poly(n, 1/\epsilon, 1/\delta)$.  

We employ \Cref{cor:perapprox}, which gives an $\FBPP^{\NP^{\mathcal O}}$ algorithm that estimates of $\left\lvert\Per X\right\lvert^2$ to within additive error
\beq
\label{eq:source_error}
 \varepsilon\frac{\norm m^n \alpha^{2n}}{|{\mathcal H_S}|}
\eeq
with probability at least $1 - \Delta$ in time $\poly(m, 1/\varepsilon, 1/\Delta)$. We will set $\Delta = \delta/2$ and set $\varepsilon$ such that Eq.~\eqref{eq:source_error} is bounded by $\epsilon n!$.  To this end, let us define the ratio
\beq
\label{eq:I_defn}
\mathcal I := \epsilon/\varepsilon = \frac{\norm \alpha^{2n} m^n}{|{\mathcal H_S}|n!}.
\eeq
Our goal now is to bound how large $\mathcal I$ can be. If it is at most polynomially large, then our proof is concluded since we can set $\varepsilon = \epsilon / \poly(n,1/\delta)$ and thus achieve the precision required by $\GPEpa$. First recall that we are working in the dilute limit, so let us set $m = c n^2$ and $\alpha^2 = \sqrt{cm}$ where $c > 1$ is some constant to be determined later.  Furthermore, in each half of the output we have $n$ photons in $m$ modes without collisions, and so
\beq
\label{eq:space_bound}
|\mathcal H_S| = \binom{m}{n}^2 \ge \left(1 - \frac{1}{c}\right)^2 \left( \frac{m^n}{n!} \right)^2 := \gamma_c \left( \frac{m^n}{n!} \right)^2
\eeq
where we have used\footnote{Whenever $c, n \ge 1$, we have
$$
\binom{m}{n} = \frac{\prod_{i=0}^{n-1} (m-i)}{n!} \ge \frac{(m-n)^n}{n!} = \frac{(1- \frac{1}{cn})^n m^n}{n!} \ge \left(1 - \frac{1}{c}\right) \frac{m^n}{n!}.
$$.}
that $m = c n^2$ and implicitly defined a new constant $\gamma_c > 0$. From this, we have

\beq
\label{eq:I_bound}
\mathcal I =  \frac{\norm \alpha^{2n} m^n}{\binom{m}{n}^2 n!}
\le \frac{\norm \alpha^{2n} n!}{\gamma_c^2 m^n}
\le \frac{2 \norm \sqrt{2 \pi n}}{\gamma_c^2} \left(\frac{\alpha^2 n}{m e}\right)^n
= \frac{2 \norm \sqrt{2 \pi n}}{\gamma_c^2} \frac{1}{e^n}
\eeq
where the first inequality uses Eq.~\eqref{eq:space_bound}, and the second uses a Stirling approximation bound. By \Cref{lem:boundZ}, we have that
\beq
\label{eq:Z_bound}
\Pr \left[ \norm \geq \frac{4 e^{272/c}}{\delta} e^{n} \right] \leq \frac{\delta}{2}
\eeq
as long as $m \ge \ln(2/\delta)$.  Once again, if $m < \ln(2/\delta)$, we get $1/\delta = \Omega(e^{c n^2})$ and so we could have computed the permanent of $X$ explicitly using Ryser's formula. Therefore, let us assume Eq.~\eqref{eq:Z_bound} and combine it with Eq.~\eqref{eq:I_bound} by the union bound to conclude that $\mathcal I = O(\sqrt{n} / \delta)$ with probability at least $1-\delta$ over the Gaussian ensemble.  This completes the proof.
\end{proof}


\section{Dealing with collision outcomes}
\label{sec:collision_reduction}
In this section, we establish connections between permanents of matrices with repeated rows and columns (corresponding to collision outcomes) and permanents of matrices without repetitions. The end goal is to establish a reduction between the hardness of computing permanents of Gaussian random matrices with and without repetitions. These results apply generally and may also be useful for BosonSampling.

To start, notice that we need to define new variants of the Gaussian Permanent Estimation problem where the goal is to estimate permanents of matrices with repeated rows and columns.  Much like how the error tolerance for GPE is based on the expectation of the permanent, the error tolerance for these new problems is based on the expectation of permanents with repetitions for which we need the following result (proof in \Cref{app:lemma_proofs}):
\begin{restatable}{lemma}{repperexpectation}
Let $A \sim \mathcal N(0, 1)^{c \times c}_\mathbb C$. Then
\begin{equation}
\mathbb E\left[\left|\Per(A_{S,T})\right|^2\right] = n! \prod_{i=1}^c s_i! \prod_{j=1}^c t_j!
\end{equation}
for any repetition vectors $S, T \in \mathbb N^c$ such that $n = \sum_{i=1}^c s_i = \sum_{j=1}^c t_j$.
\end{restatable}
We can now define the following problems that generalize $\GPEpa$ and $\GPEa$:
\begin{definition}\label{def:RGPE}
Given $A \sim \mathcal N(0, 1)^{c \times c}_\mathbb C$, vectors $S, T \in \mathbb{N}^c$ with $n = \sum_{i=1}^c s_i = \sum_{j=1}^c t_j$, and accuracy parameters $\varepsilon,\delta>0$, we define the following problems for which we are required to output a complex number $\tilde{P}$ such that
\begin{enumerate}
    \item $\RGPEpa : \left|\tilde{P}-\left|\Per(A_{S,T})\right|^2\right|\leq \varepsilon n! \prod_{i,j=1}^c s_i! t_j!$,
    \item $\RGPEa : \left|\tilde{P}-\Per(A_{S,T})\right|\leq \varepsilon \sqrt{n! \prod_{i,j=1}^c s_i! t_j!}$,
\end{enumerate}
with probability at least $1-\delta$ over the randomness of $A$ in time $\poly(n, c, 1/\varepsilon, 1/\delta)$.\footnote{Note that we have chosen to define $A$ as a $c \times c$ Gaussian matrix, rather than an $m \times m$ matrix. This choice was motivated by the fact that the complexity of the RGPE problem will depend on the number of total clicks in the repetition patterns, rather than the total number of modes. In fact, this idea will be important later in \Cref{thm:PerAsPoly}, where we require that there is at least one photon in every mode ($s_i,t_j \ge 1$).}
\end{definition}

Given that we do not have a conclusive proof of hardness in the high-collision regime, some may wonder whether we have defined $\RGPEpa$ correctly. Do we need to make the error scale with the expectation as is the case with $\GPEpa$? In \Cref{sec:beyond_dilute}, we give numerical evidence that this is the right choice---i.e., Stockmeyer counting on the BipartiteGBS distribution in the high-collision regime still solves the $\RGPEpa$ problem. Moreover, if the error scaling were much smaller, then the estimate it produced would be insufficient. 

In \cite{aaronson2013}, the authors show that $\GPEpa$ and $\GPEa$ are, in fact, equivalent, up to the Permanent Anti-Concentration Conjecture (PACC), which states that the permanents of random Gaussian matrices are not too concentrated around 0.\footnote{Formally, the Permanent Anti-Concentration Conjectures states that there is some polynomial $p$ such that for all $c$ and $\delta > 0$, $\Pr_{A \sim \mathcal N(0, 1)^{c \times c}_\mathbb C}[|\Per(A)|^2 < c!/p(c, 1/\delta)] < \delta$.} With some amount of work, it is possible to prove a similar equivalence between problems $\RGPEpa$ and $\RGPEa$ assuming a version of the PACC which includes permanents of Gaussian matrices with repetitions. Unfortunately, even this generalization has its limits. It can be shown that, for a matrix composed of $c$ copies of a single row of Gaussian elements, the permanent does not anti-concentrate. Therefore, even if the PACC is true, we expect its generalization to fail above a certain number of row/column repetitions. The hope is that the generalized PACC might nonetheless hold for those repetition patterns which occur naturally in the GBS distribution. We do not pursue this here as this equivalence is not required for the arguments that follow.

This rest of this section will show how solutions to RGPE can be used to solve GPE.  In \Cref{sec:permanentreduc}, we show how the permanent of a matrix with repetitions can be written as a polynomial that depends on the permanent of the matrix without repetitions. In \Cref{sec:poly_interpolation}, we use this result to prove an efficient reduction between $\RGPEa$ and $\GPEa$ and an inefficient reduction between $\RGPEpa$ and $\GPEpa$. The latter result shows that BipartiteGBS is hard when there are $O(1)$ collisions in the outcomes.

\subsection{Permanents of matrices with repeated rows and columns} \label{sec:permanentreduc}
Given a matrix $A$ and repetition patterns $S,T$, our goal for this section is to construct some matrix $B$ such that $\Per(B_{S,T})$ can be expressed as a polynomial whose constant coefficient is proportional to $\Per(A)$.  Therefore, given an oracle to compute the permanents of matrices with this collision pattern, we can infer (through polynomial interpolation) the permanents of matrices that have no repeated rows or columns.  The degree of the polynomial will depend on the total number of collisions $k := k_S + k_T$, i.e., the number of collisions in repetition patterns $S$ and $T$, respectively.  Note that here we are allowing $k_S \neq k_T$.  If this is the case, the matrix $B$ must be rectangular so that the matrix $B_{S,T}$ is square (for which the permanent is well-defined).  Formally, we show the following:
\begin{theorem}\label{thm:PerAsPoly}
Let $A \in \mathbb C^{c\times c}$ and $S,T \in \mathbb N^c$ such that $s_i,t_j \ge 1$ be given.  There is a $(c + k_T) \times (c + k_S)$ matrix $B := B[z;x,y]$ obtained by adding $k_S$ columns and $k_T$ rows to the matrix $A$ with entries that depend on complex variables $z$, $x = \{x_{i,j}^{(\ell)}\}$, and $y = \{y_{i,j}^{(\ell)}\}$ such that
\begin{align}
\Per(B_{S,T}) &=  \xi\Per(A) + \gamma_1 z + \gamma_2 z^2 + \ldots + \gamma_k z^k \label{eq:PerAsPoly} \\
\xi &= \prod_{i=1}^cs_i!t_i! \left(\prod_{j=1}^{t_i -1}x^{(i)}_{i, j} \right) \left(\prod_{j=1}^{s_i -1}y^{(i)}_{i, j} \right). \label{eq:xi_def}
\end{align}
The coefficients $\gamma_1,\ldots, \gamma_k$ are functions of the entries of $A$, $S$, $T$, $x$ and $y$. Furthermore, whenever $x_{i,j}^{(\ell)}$, $y_{i,j}^{(\ell)}$ and all entries of $A$ are i.i.d.\ standard complex Gaussians, and $|z| = 1$, we have that $B \sim \mathcal N(0, 1)^{(c + k_T) \times (c + k_S)}_\mathbb C$.
\end{theorem}

While the majority of this proof is given \Cref{app:lemma_proofs}, let us briefly describe the form of the construction, which results from two separate (but essentially identical) steps.  We first consider the case where only the rows of $A$ are repeated.  That is, we construct a rectangular matrix $B'$ such that when \emph{only} the row repetition pattern $S$ is applied (yielding a square matrix $B_S'$) the analogous version of Eq.~\eqref{eq:PerAsPoly} holds (i.e., the constant coefficient of the permanent polynomial is proportional to the permanent of $A$).  In the second step, we simply treat $B_S'$ as if it were the original matrix $A$ and do the same construction except for the repeated columns specified by the repetition pattern $T$.  

Let us now describe this first step, the construction of $B'$.  For each $\ell$ such that $s_{\ell} > 1$, append to $A$ the $c\times (s_\ell-1)$ matrix $V^{(\ell)}$ with entries
\beq
V^{(\ell)}_{ij} =
\begin{cases}
y_{i,j}^{(\ell)} & \text{ if } i = \ell \\
z y_{i,j}^{(\ell)} & \text{ o.w. }
\end{cases}
\eeq
For example, the $2\times 2$ matrix $A = \begin{pmatrix}
a_{1,1} & a_{1,2}  \\
a_{2,1} & a_{2,2}
\end{pmatrix}$ with repetition vector $S=(3,2)$ has
\begin{align*}
    V^{(1)} = \begin{pmatrix}
    y^{(1)}_{1, 1} & y^{(1)}_{1, 2} \\
    z\, y^{(1)}_{2, 1} & z\,y^{(1)}_{2, 2}
    \end{pmatrix} \hspace{0.5cm}  \text{ and } \hspace{0.5cm} V^{(2)} = \begin{pmatrix}
    z\, y^{(2)}_{1, 1}  \\
    y^{(2)}_{2, 1}
    \end{pmatrix},
\end{align*}
which leads to the matrix
\begin{align*}
B' =
   \begin{pmatrix}
   A & | & V^{(1)} & | & V^{(2)}
   \end{pmatrix} = \begin{pmatrix}
    a_{1,1} & a_{1,2} & y^{(1)}_{1, 1} & y^{(1)}_{1, 2} & z\, y^{(2)}_{1, 1} \\
    a_{2,1} & a_{1,2} & z\, y^{(1)}_{2, 1} & z\,y^{(1)}_{2, 2} & y^{(2)}_{2, 1}
    \end{pmatrix}.
\end{align*}
When we apply the row repetition pattern $S$, we get
\begin{align*}
    B'_S = \begin{pmatrix}
    a_{1,1} & a_{1,2} & y^{(1)}_{1, 1} & y^{(1)}_{1, 2} & z\, y^{(2)}_{1, 1} \\
    a_{1,1} & a_{1,2} & y^{(1)}_{1, 1} & y^{(1)}_{1, 2} & z\, y^{(2)}_{1, 1} \\
    a_{1,1} & a_{1,2} & y^{(1)}_{1, 1} & y^{(1)}_{1, 2} & z\, y^{(2)}_{1, 1} \\
    a_{2,1} & a_{2,2} & z\, y^{(1)}_{2, 1} & z\,y^{(1)}_{2, 2} &  y^{(2)}_{2, 1} \\
    a_{2,1} & a_{2,2} & z\, y^{(1)}_{2, 1} & z\,y^{(1)}_{2, 2} &  y^{(2)}_{2, 1}
    \end{pmatrix}.
\end{align*}
To give some intuition for the proof of correctness, one can show that each monomial in the expansion of the permanent for $B'_S$ cannot have two $a_{i,j}$ terms from the same row or column of $A$ unless there is also $z$ term.  Therefore, the only monomials which do not have a factor of $z$ are those that also appear in the expansion for the permanent of $A$ (albeit with an extra factor of a product of $y$ variables which is always the same).  To complete the proof, it suffices to count the multiplicity of these monomials.

This result can now be used to show that the permanent of a matrix without repetitions can be estimated from estimates of larger matrices with repetitions. The first straightforward way of doing this is just setting $z = 0$ and $x_{i,j}^{i} = y_{i,l}^{i} = 1$ in our construction. This can be viewed as a ``worst-case reduction'' between these two problems, in the sense that we are assuming we have full control over the choice of the larger matrix $B$. However, this is not sufficient to give a reduction between RGPE and GPE since the matrices there are required to have independent (up to the corresponding repetitions) Gaussian matrix elements, and the matrix $B[0; x, y]$ is very far from Gaussian.  The next section describes two (related) techniques for dealing with \Cref{eq:PerAsPoly} when the matrices are random.

\subsection{Polynomial interpolation techniques to estimate Gaussian permanents}
\label{sec:poly_interpolation}

The first polynomial interpolation technique to deal with \Cref{eq:PerAsPoly} is to query the polynomial at the $(k+1)$th roots of unity. This leads to the error scaling described in \Cref{thm:reduction_intro} in the Introduction, which we rephrase below (proof in \Cref{app:lemma_proofs}):
\begin{restatable}{theorem}{GPEaToRGPEa}\label{thm:GPEa_to_RGPEa}
Given access to an oracle for $\RGPEa$ with error $\varepsilon \sqrt{(c+k)! \prod s_i! t_j!}$, it is possible to use $k+1$ calls to the oracle to solve $\GPEa$ with error $\epsilon \sqrt{c!}$ where $\epsilon := \varepsilon(\frac{1}{|\xi|} \sqrt{((c+k)!/c!)\prod s_i! t_j!})$ with $\xi$ as in Eq.~\eqref{eq:xi_def}.
\end{restatable}
Notice that if we disregard the need to convert between the allowable $(\sqrt{(c+k)! \prod s_i! t_j!})$-error fluctuations for $\RGPEa$ and the $(\sqrt{c!})$-error fluctuations for $\GPEa$, then this theorem implies an error of $\epsilon$ for the repeated matrix leads to an error of $\epsilon/|\xi|$ for the unrepeated matrix.  Moreover, we have that $\epsilon/|\xi| = O(\epsilon/1.498^k)$ with high probability (also proved in \Cref{app:lemma_proofs}), and so we actually obtain an exponential improvement in the error accuracy.  Unfortunately, once the target error bounds for $\RGPEa$ and $\GPEa$ \emph{are} incorporated we are in opposite situation---$\epsilon := \varepsilon(\frac{1}{|\xi|} \sqrt{((c+k)!/c!)\prod s_i! t_j!})$ is exponentially large.  Therefore, \Cref{thm:GPEa_to_RGPEa} cannot be used to reduce $\GPEa$ to $\RGPEa$ without an exponential error blowup.

That said, we believe that \Cref{thm:GPEa_to_RGPEa} provides weak evidence of a formal connection between $\RGPEa$ and $\GPEa$ (and recall the latter is believed to be $\#\P$-hard).  Furthermore, it is worth noting that \Cref{thm:GPEa_to_RGPEa} leads to a somewhat surprising error scaling, as the error in the polynomial interpolation does not depend explicitly on the polynomial degree $k$.   We hope the relatively benign error scaling of \Cref{thm:GPEa_to_RGPEa} might make it a useful building block in a hardness proof for GBS or BosonSampling in the regime where high-collision outcomes dominate, and we leave filling the above gaps for future research.

However, a crucial aspect of this reduction in \Cref{thm:GPEa_to_RGPEa} is that it is between \emph{amplitudes} (whose phases are needed to arrange the cancellation of unwanted terms), whereas hardness arguments must be stated in terms of probabilities. We now describe an alternative polynomial interpolation technique which does provide an efficient reduction from $\GPEpa$ to $\RGPEpa$, but only in the regime where $k = O(1)$.  While the scaling in \Cref{thm:GPEa_to_RGPEa} will be much better and would also suffice in the constant-collision regime, we note that working directly with permanent magnitudes allows us to avoid creating a new anti-concentration conjecture for permanents with repeated rows and columns---i.e., to reduce from $\RGPEa$ to $\RGPEpa$:

\begin{restatable}{theorem}{GPEpaToRGPEpa} \label{thm:GPEpa_to_RGPEpa}
Given access to an oracle for $\RGPEpa$ with error $\varepsilon (c+k)! \prod s_i! t_j!$, it is possible to use $O(k)$ calls to the oracle to solve $\GPEpa$ with additive error $\epsilon c!$, for
\beq
\epsilon := \Ord\left(\frac{(c+k)^{2k}k^{2k}}{\delta^{3k+{1}/{2}}2^k} \varepsilon \right)
\eeq
Whenever $k = O(1)$, we have $\epsilon = \poly(c,\varepsilon, \delta)$.
\end{restatable}

The main reason why the error above scales so poorly in $k$ is that we take the absolute squared value of Eq.~\eqref{eq:PerAsPoly}, and then apply a polynomial interpolation method. Contrary to \Cref{thm:GPEa_to_RGPEa}, we cannot compute the polynomial for values of $z$ in the unit circle, since many of its monomials depend on $\left|z\right|^2$ and so we cannot orchestrate the proper cancellations. As a result, we use a least squares estimator to approximate the value of the polynomial at $z=0$ by sampling $O(k)$ values of it in a small range around $z=1$. This causes a blowup that is exponential in the degree of the polynomial, leading to the scaling above.  

In \Cref{app:lemma_proofs} we outline the proof of this result, though we omit some of the important technical lemmas that can be taken directly from \cite{AaronsonBrod} with regards to the hardness of lossy BosonSampling.

\Cref{thm:GPEpa_to_RGPEpa} can be plugged, in a relatively direct manner, into our main argument or that of \cite{aaronson2013} to prove that the complexity of both BosonSampling and GBS remains unchanged when we move from the no-collision subspace to a subspace with a (fixed) constant number of collisions. While this does not provide evidence of hardness in a regime where $m = O(n^d)$ for $d<2$, we believe \Cref{thm:GPEpa_to_RGPEpa} is well-motivated in two senses: first, as a preliminary step towards a stronger result that does take into account more that constantly-many collisions, which might re-purpose some of the intermediate results; second, as a level of ``experimental robustness'', where allowing for a few collision outcomes to be included in the output distribution might improve the count-rate for finite-sized experiments (recall that even in the no-collision regime, some constant-fraction of the output distribution may still have collisions).


\section{Discussion} \label{subsec:discussion}
We introduced BipartiteGBS as a method for programming a Gaussian boson sampling device so that the output probabilities are proportional to the permanents of arbitrary matrices.  This allowed us to rigorously prove the hardness of approximately sampling from GBS distributions in the $m = \Theta(\braket n^2)$ regime under the same set of conjectures as those used for BosonSampling~\cite{aaronson2013}.  To recap the advantage of our approach, recall the two reasons why BosonSampling is required to operate in the $m=\Omega(n^2)$ (dilute) regime. The first is that the argument is predicated on the hardness of permanents with unrepeated rows and columns, and so we need that many more modes than photons for no-collision outcomes to dominate the output distribution.  Perhaps more importantly, the second reason is that it is required that the $n \times n$ submatrices of $m \times m$ Haar-random matrices appear approximately Gaussian.  In fact, while it is widely believed that such a statement holds whenever $m = \omega(n^2)$, it was only formally proved for $m=\omega(n^5)$. An important aspect of our work is that it has removed this obstacle in the case of GBS. Since we can implement arbitrary transition matrices, we can just choose them directly from the Gaussian ensemble, and thus their submatrices already look Gaussian even for $m=\Theta(n)$. Thus, there is the tantalizing prospect of improving on our work to show hardness for GBS in this regime, which would be much more experimentally friendly.  We provide a blueprint for how this argument might go in \Cref{sec:beyond_dilute}.

Let's also contrast our work with another recent paper of Deshpande et al.\ on the hardness of Gaussian boson sampling~\cite{deshpande2021quantum}.  There, the authors also prove a worst-to-average case hardness result for approximate GBS in the dilute limit (c.f.~\Cref{thm:main}).  To do this, they must conjecture two plausible and yet new conjectures in complexity and random matrix theory.  Furthermore, Deshpande et al.\ actually inherit the same problem that appears in BosonSampling---there are only rigorous proofs that the $n \times n$ submatrices of $m \times m$ are approximately Gaussian when $m = \omega(n^5)$.  To circumvent this problem, the authors must conjecture directly that $m = \Omega(n^2)$ suffices (technically, they require that submatrices of the unitary product $UU^T$ approximate random $X X^T$ matrices for Gaussian $X$).  Therefore, they will require fundamentally new ideas to operate in regimes beyond the dilute limit.

We reiterate that the main open problem left by our work is exactly that---prove hardness of GBS in a regime where the number of modes is subquadratic in the number of photons.  We outline two possible approaches.  First, one could attempt to improve the reduction in \Cref{sec:collision_reduction} so that an efficient algorithm for $\RGPEa$ implies an efficient algorithm for $\GPEa$.  This would allow for the possibility that the hardness of GBS in the high-collisions could primarily be based on the same hardness conjectures as those used in BosonSampling.  That said, several issues still need to be addressed.  For example, consider our bound on $\mathcal Z$ in \Cref{lem:boundZ}.  There, we have a multiplicative $\exp(\Omega(m/\alpha^4))$ term which is constant in the dilute limit, but grows exponentially for smaller $\alpha$.  Some change in scaling is inevitable, but this bound will not lead to a tight enough estimation of the permanent to solve $\RGPEa$ via our Stockmeyer counting argument.  While we give strong evidence in \Cref{sec:beyond_dilute} that such a bound exists, proving it rigorously may be challenging.  Moreover, in order to use a reduction akin to that given in \Cref{thm:GPEa_to_RGPEa}, one would need to prove an equivalence between the multiplicative and additive versions of RGPE. This would required a generalization of the Permanent Anti-Concentration Conjecture for matrices with repeated rows and columns which is not known to be equivalent to the original conjecture, and in fact is provably false in extreme settings. A second approach would be to try to strengthen the evidence for the $\#\P$-hardness of the Permanent-of-Repeated Gaussians problem without reference to the BosonSampling conjectures. We note that even in this case much of the same work outlined above would be required.

Another major open direction left in our work---and indeed in many of the proposed hardness arguments based on linear optics---is incorporating realistic experimental noise. To be clear, we show the classical hardness of approximately sampling (in total variation distance) from ideal GBS distributions. However, even in the high-collision regime, current experiments \emph{also} do not sample approximately from their true distributions due to a variety of sources of noise (e.g., photon loss). Therefore, incorporating this noise is a critical to closing the gap between theory and experiment.

That said, it is likely that the flexibility of BipartiteGBS can help mitigate the effects of noise, particularly in near-term experiments. For instance, assume one wants to perform Scattershot BosonSampling with an unitary transition matrix $W$, which corresponds to BipartiteGBS where we set all squeezing parameters to be equal, $U=W$, and $V = \mathbb{I}$ (cf.\ \Cref{Fig: GBS}). If $W$ is Haar random, then it typically requires depth equal to $m$ \cite{reck1994,clements2016}. But now note that we can shift some of the beam splitters from $U$ to $V$, in the sense that any combination of these two matrices for which $W = U V^T$ leads to the same problem instance. Thus, we can program an $m$-depth decomposition of $W$, such as that of \cite{reck1994,clements2016}, in an $m/2$-depth BipartiteGBS instance simply by mapping half of the beam splitter layers to $U$ and half to $V$. Since losses scale exponentially with depth, this reduction of a factor of two for the depth corresponds to having square root of the losses.

Another less straightforward example is as follows. Suppose we measure the complexity of simulating the device by the computational cost of producing a single $2n$-photon sample (from the exact distribution) using state-of-the-art algorithms. A $2m$-mode implementation of the standard GBS model would typically require an arbitrary 2$m$-mode interferometer, and samples of its output probabilities can be generated in time $O(m n^3 2^n)$~\cite{bjorklund2018faster,quesada2022quadratic,bulmer2022boundary}. A BipartiteGBS instance with comparable computational cost (up to polynomial factors) would also have $2m$ total modes, and its $2n$-photon probabilities would be given by permanents of $n \times n$ matrices, with cost-per-sample of $O(n 2^n)$ time~\cite{ryser1963combinatorial,neville2017classical,clifford2018classical,clifford2020faster}. But an arbitrary BipartiteGBS interferometer only requires depth $m$, and thus half the depth (and square root of the losses) as an arbitrary implementation of standard GBS. Combining these two observations, it is possible to reduce the depth by a factor of 4 if one moves from arbitrary GBS to an implementation of BipartiteGBS based on an unitary transition matrix.

Finally, we ask what other uses BipartiteGBS might have in proving stronger forms of classical intractability---either from an experimental or theoretical viewpoint.  For instance, because we are no longer constrained to have unitary transition matrices, one might envision a hardness result predicated on an entirely different distribution of matrices (e.g.\ Bernoulli instead of Gaussian).  The flexibility of BipartiteGBS allows the possibility of more contrived distributions that nevertheless have a stronger complexity-theoretic foundation, and we hope that BipartiteGBS motivates others to explore these possibilities.


\section*{Acknowledgments}
We would like to thank Torben Kr\"uger for useful discussions.
N.Q. thanks Z. Vernon for valuable discussions and the Ministère de l’Économie et de l’Innovation du Québec and the Natural Sciences and Engineering Research Council of Canada for financial support.
D.\ J.\ Brod acknowledges support from Instituto Nacional de Ci\^{e}ncia e Tecnologia de Informação Quântica (INCT-IQ, CNPq) and FAPERJ. M.\ B.\ A.\ Alonso acknowledges support from CNPq.
\bibliographystyle{quantum}
\bibliography{references}

\appendix


\section{Bounds on \texorpdfstring{$n$}{photon number} and \texorpdfstring{$\norm$}{normalization term} from random matrix theory} \label{app:bounds}

In this appendix, we prove \Cref{lem:boundn} and \Cref{lem:boundZ}, corresponding to bounds on $n$ and $\norm$, respectively. As in the main text, let $\mathcal G = \mathcal N(0, \tfrac{1}{\alpha^2 m})^{m \times m}_\mathbb C$ be the distribution over $m \times m$ matrices whose entries are independent complex Gaussians with mean $0$ and variance $1/\alpha^2 m$.

The quantities we will bound depend crucially on the set $\{\lambda_i\}_{i=1,\ldots,m}$, where $\lambda_i = \sigma_i^2$ and $\sigma_i$ is the $i$th singular value of Gaussian matrix $C \sim \G$. Alternatively, $\lambda_i$ is the $i$th eigenvalue of the matrix $A := CC^\dag$. The matrix $A$ is called a \emph{complex Wishart matrix}. The complex Wishart distribution is usually denoted by $\mathcal W(\mu, \Sigma)$ when the columns of $C$ are sampled from the multivariate complex normal distribution with mean vector $\mu$ and covariance matrix $\Sigma$. In our case we can write $A \sim \mathcal W(0, \frac{1}{\alpha^2 m} \mathbb{I}_m)$, and the latter ensemble will be represented by $\mathcal W$, for short.

Much is known about the spectrum of complex Wishart matrices from random matrix theory, and we will invoke several known results regarding $\mathcal W$ throughout this appendix. For example, it is known that the maximum eigenvalue of $A\sim \mathcal W$ (denoted by  $\lambda_{\max}(A)$) can be bounded as follows:
\begin{lemma}[Haagerup and Thorbj{\o}rnsen \cite{haagerup_thorbjrnsen:2003}]
\label{lem:max_eigen_wishart}
\beq
\Pr[\lambda_{\max}(A) \ge 4/\alpha^2 + \epsilon] \le m e^{-m \alpha^4 \epsilon^2 / 8}.
\eeq
\end{lemma}

Let us recall the notation we will use for the source of randomness in these calculations. We write $\braket{\cdot}$ for the average of a quantity (e.g., the total photon number) over the distribution of photon numbers that arises in a GBS setup of \Cref{subsec:model}, for a particular transition matrix $C$. We write $\mathbb E[\cdot]$ as the expectation of a particular quantity over the randomness of $C \sim \G$.

\subsection{Bounds on observed photon number} \label{subapp:boundn}

We wish to upper bound the probability that the observed  number of photon pairs, $n$, is far from its expectation. Let us begin by reviewing the computation of the expectation value itself. We can write
\beq
\langle n \rangle = \sum_{i=1}^{m} \frac{\lambda_i}{1- \lambda_i}
=  \sum_{i=1}^{m} \sum_{k=1}^\infty \lambda_i^k
=  \sum_{k=1}^\infty \Tr(A^k),
\eeq
where $\lambda_i$ is the $i$th eigenvalue of matrix $A \sim \mathcal W$.  We can compute the exact expected value of $\langle n \rangle$ over random $A$ using the following lemma:

\begin{lemma}[Hanlon, Stanley, Stembridge \cite{hss:1992_WishartTrace}]
For $A \sim \mathcal W(0, \sigma^2 \mathbb{I}_m)$,
\beq
\mathbb E[\Tr(A^k)] = \frac{\sigma^{2k}}{k} \sum_{i=1}^k (-1)^{i-1} \frac{(m+1-i)_k^2}{(k-i)!(i-1)!},
\eeq
where $(a)_k := a(a+1)\cdots(a + k - 1)$ is the rising Pochhammer symbol.
\end{lemma}

Therefore, for $A \sim \mathcal W(0, \frac{1}{\alpha^2 m} \mathbb{I}_m)=\mathcal W$, the first several instances (starting at $k=1$) of $\mathbb E[\Tr(A^k)]$ are
\beq
\frac{m}{\alpha^2},\frac{2 m}{\alpha^4},\frac{5 m^2+1}{\alpha^6 m},\frac{14 m^2+10}{\alpha^8 m},\frac{42 m^4+70 m^2+8}{\alpha^{10} m^3}, \ldots
\eeq
This makes the expression for $\mathbb E [\langle n \rangle ]$ challenging to write down and manipulate. Notice, however, that when $\alpha$ is large, the higher order terms are negligible.  Formally, we state the following proposition:
\begin{proposition}
\label{prop:<n>_lambda_max}
$
\langle n \rangle \le \Tr(A) + 2 m \lambda_{\max}(A)^2
$
whenever $\lambda_{\max}(A) \le 1/2$.
\end{proposition}
\begin{proof}
Start by writing $\langle n \rangle =  \sum_{i=1}^{m} \sum_{k=1}^\infty \lambda_i^k$.  We can bound all the higher-order terms ($k > 1$) as
\beq
\sum_{i=1}^{m} \sum_{k=2}^\infty \lambda_i^k
\le m \sum_{k=2}^\infty \lambda_{\max}(A)^k
\le 2 m \lambda_{\max}(A)^2,
\eeq
whenever $\lambda_{\max}(A) \le 1/2$.  What remains (the $k=1$ term) is simply the trace of $A$.
\end{proof}

Importantly, we know that the maximum eigenvalue of $A \sim \mathcal W$ is small (for large $\alpha$) with high probability, due to \Cref{lem:max_eigen_wishart}. We also know the expectation of the trace:
\begin{lemma}[Maiwald and Kraus \cite{maiwaldkraus:2000_WishartTrace}]
\label{lem:trace_wishart}
 $\mathbb E[\Tr(A)] = \frac{m}{\alpha^2}$
and $\mathbb E[\Tr(A)^2] = \frac{m^2+1}{\alpha^4}$.
\end{lemma}

We arrive at the following conclusion about the distribution of $\langle n \rangle$ over the choice of $A$:
\begin{theorem}
\label{thm:<n>_variance}
\beq
\Pr\left[\left|\langle n \rangle - \frac{m}{\alpha^2}\right| \ge \frac{32 \beta^2 m}{\alpha^4} + \frac{1}{\alpha^2}\sqrt{\frac{2}{\delta}}\right] \le \delta,
\eeq
whenever $\beta \ge 4$, $\alpha^2 \ge 8 \beta$, and $m \ge \beta^{-2} \ln(1/\delta)$.
\end{theorem}
\begin{proof}
By Lemma~\ref{lem:max_eigen_wishart}, for any $\beta \ge 4$, we have
\beq
\Pr[\lambda_{\max}(A) \ge 4\beta/\alpha^2] \le m e^{-2m(\beta-1)^2} \le \frac{\delta}{2},
\eeq
whenever $m \ge \beta^{-2} \ln(1/\delta)$, so let us assume $\lambda_{\max}(A) \le 4\beta/\alpha^2 \le 1/2$.
Let us write $\langle n \rangle = \Tr(A) + \xi(A)$ where $\xi(A) := \sum_{k=2}^\infty \Tr(A^k)$.  Using Proposition~\ref{prop:<n>_lambda_max} and the fact that $A$ is positive semidefinite, we have that $0 \le \xi(A) \le 32 \beta^2 m/\alpha^4$.

Therefore, to understand the deviation of $\langle n \rangle$, it suffices to consider the variance of $\Tr(A)$.  By  Lemma~\ref{lem:trace_wishart}, we have $\Var[\Tr(A)] = \mathbb E[\Tr(A)^2] - \mathbb E[\Tr(A)]^2 = 1/\alpha^4$, and so, by Chebyshev's inequality, we get
\beq
\Pr\left[|\Tr(A) - \mathbb E[\Tr(A)]|
\ge \frac{1}{\alpha^2}\sqrt{\frac{2}{\delta}} \right] \le \frac{\delta}{2}.
\eeq
Combining our two bounds with the union bound, we get
\begin{align}
|\langle n \rangle - \mathbb E[\Tr(A)]|
&\le |\langle n \rangle - \Tr(A)| + |\Tr(A) - \mathbb E[\Tr(A)]| \nonumber\\
&\le |\xi(A)| + \frac{1}{\alpha^2}\sqrt{\frac{2}{\delta}} \nonumber\\
&\le \frac{32 \beta^2 m}{\alpha^4} + \frac{1}{\alpha^2}\sqrt{\frac{2}{\delta}},
\end{align}
which (using $\mathbb E[\Tr(A)] = m/\alpha^2$) gives the theorem.
\end{proof}

Therefore, assuming $n$ is roughly equal to its expected value (over both the choice of $A$ and the inherent randomness of the photon number), we can approximate $n$ as $m/\alpha^2$ whenever $\alpha = \Omega(m^{1/4})$ from Proposition~\ref{prop:<n>_lambda_max}.  What we would like to know is how much $n$ deviates from this value, for which we will need to compute the variance
\beq
\Delta^2 [n] = \langle n^2 \rangle -  \langle n \rangle^2
\eeq

Let us now write $n = \sum_{i=1}^m n_i$ where $n_i$ is the number of photons generated at source in mode $i$.  Expanding out the $\langle n^2 \rangle$ term of the variance, we get
\beq
\langle n^2 \rangle
= \sum_{i=1}^m \langle n_i^2 \rangle + \sum_{i \neq j} \langle n_i \rangle \langle n_j \rangle
= \langle n \rangle^2 + \langle n \rangle + \sum_{i=1}^m \braket{n_i}^2
\eeq
where the first equality uses $\braket{n_i n_j} = \braket{n_i}\braket{n_j}$ (i.e., the photons are generated independently at each mode), and the second equality uses the following identity for one half of a two-mode squeezed state:\footnote{See, for example, Eq.~(3.5.15) of Barnett and Radmore~\cite{barnett2002}.  For completeness, we give a self-contained proof of this fact in \Cref{sec:thermal_state_identity}.}
\beq
\braket{n_i^2} = 2\braket{n_i}^2+\braket{n_i}.
\eeq
Therefore, we can write the variance as
\beq
\Delta^2 [n] = \langle n \rangle + \sum_{i=1}^m \left(\frac{\lambda_i}{1- \lambda_i}\right)^2 = \langle n \rangle + \sum_{i=1}^m \sum_{k=2}^\infty (k-1)\lambda_i^k.
\eeq
Once again, we make the observation that when the maximum eigenvalue of $A$ is small, the higher order terms in the Taylor expansion converge.

\begin{proposition}
\label{prop:var(n)_lambda_max}
$
\Delta^2 [n] \le \langle n \rangle + 4 m \lambda_{\max}(A)^2,
$
whenever $\lambda_{\max}(A) \le 1/2$.
\end{proposition}
\begin{proof}
We have $\Delta^2 [n] = \langle n \rangle + \sum_{i=1}^m \sum_{k=2}^\infty (k-1)\lambda_i^k$.  We bound the terms of the sum as
\beq
\sum_{i=1}^m \sum_{k=2}^\infty (k-1)\lambda_i^k
\le m \sum_{k=2}^\infty (k-1)\lambda_{\max}(A)^k
\le 4 m \lambda_{\max}(A)^2
\eeq
whenever $\lambda_{\max}(A) \le 1/2$.
\end{proof}

We are finally ready to state and prove the main result for this section: a bound on how much $n$ deviates from $m/\alpha^2$ (roughly, its expected value) over both the randomness of $A \sim \mathcal W$ and the randomness of the photon number.

\begin{theorem} \label{thm:n_variance}
\beq
\Pr\left[ \left|n - \frac{m}{\alpha^2}\right| \ge \frac{2 \sqrt{m}}{\alpha \sqrt{\delta}} +  \frac{3}{\alpha \delta^{3/4}} + \frac{21 \beta \sqrt{m}}{\alpha^2\sqrt{\delta}} + \frac{32 \beta^2 m}{\alpha^4} \right] \le \delta,
\eeq
whenever $\beta \ge 4$, $\alpha^2 \ge 8 \beta$, and $m \ge 1/\beta^2 \ln(1/\delta)$.
\end{theorem}
\begin{proof}
Once again, we have $\lambda_{\max}(A) \le 4\beta/\alpha^2 \le 1/2$ with probability $\delta/4$ by Lemma~\ref{lem:max_eigen_wishart}, and so we have
\beq
\Delta^2 [n] \le \langle n \rangle + \frac{64 \beta^2 m}{\alpha^4}
\le \frac{m}{\alpha^2} + \frac{2}{\alpha^2\sqrt{\delta}} + \frac{96 \beta^2 m}{\alpha^4}
\eeq
where the first inequality comes from Proposition~\ref{prop:var(n)_lambda_max} and the second comes from Theorem~\ref{thm:<n>_variance} (which we assume holds with probability $\delta/2$).  By Chebyshev's inequality, we get
\beq
\Pr\left[|n - \langle n \rangle| \ge \sqrt{\frac{4 \Delta^2 [n]}{\delta}} \right] \le \frac{\delta}{4},
\eeq
where the probability is over the randomness in the photon number.  And so, combining the previous inequalities with the union bound, we get
\begin{align}
|n - m/\alpha^2| &\le |n - \langle n \rangle| + |\langle n \rangle - m/\alpha^2|\nonumber \\
&\le \sqrt{\frac{4\Delta^2 [n]}{\delta}} + \frac{2}{\alpha^2\sqrt{\delta}} + \frac{32 \beta^2 m}{\alpha^4} \nonumber\\
&\le \frac{2}{\sqrt{\delta}}\sqrt{\frac{m}{\alpha^2} + \frac{2}{\alpha^2\sqrt{\delta}} + \frac{96 \beta^2 m}{\alpha^4}} + \frac{2}{\alpha^2\sqrt{\delta}} + \frac{32 \beta^2 m}{\alpha^4} \nonumber\\
&\le \frac{2}{\sqrt{\delta}}\left(\sqrt{\frac{m}{\alpha^2}} + \sqrt{\frac{2}{\alpha^2\sqrt{\delta}}} + \sqrt{\frac{96 \beta^2 m}{\alpha^4}}\right) + \frac{2}{\alpha^2\sqrt{\delta}} + \frac{32 \beta^2 m}{\alpha^4} \nonumber\\
&\le \frac{2 \sqrt{m}}{\alpha \sqrt{\delta}} +  \frac{3}{\alpha \delta^{3/4}} + \frac{21 \beta \sqrt{m}}{\alpha^2\sqrt{\delta}} + \frac{32 \beta^2 m}{\alpha^4},
\end{align}
with probability $\delta$, which completes the proof.
\end{proof}

\subsection{Bounds on normalization factor} \label{subapp:boundN2}

We are now interested in the quantity
\beq
\norm = \frac{1}{\det(\mathbb{I}_m - CC^\dag)} = \prod_{i=1}^{m} \frac{1}{1- \lambda_i}
\eeq
where $\lambda_i$ is the $i$th eigenvalue of a complex Wishart matrix $A := CC^\dag \sim \mathcal W$. We begin this section with the following fact about the moment generating function of $A \sim \mathcal W$:
\begin{lemma}[Goodman \cite{goodman:1963_wishart}]
\label{lem:moment_generating_wishart}
\beq
\mathbb E[\exp(\Tr(A))] = \left(1-\frac{1}{\alpha^2 m}\right)^{-m^2}.
\eeq
\end{lemma}

Our goal for this section is to prove the following:
\begin{lemma}
\beq
\Pr\left[\norm \ge  \frac{2}{\delta} e^{m/\alpha^2} e^{17\beta^2 m / \alpha^4}\right] \le \delta,
\eeq
whenever $\beta \ge 4$, $\alpha^2 \ge 8 \beta$, and $m \ge 1/\beta^2 \ln(1/\delta)$.
\end{lemma}
\begin{proof}
Recall that $\norm =  \prod_{i=1}^{m} (1- \lambda_i)^{-1}$. We have
\begin{align}
\label{eq:log_expansion}
 \prod_{i=1}^{m} (1- \lambda_i)^{-1}
= \exp\left(- \sum_{i=1}^{m} \ln(1- \lambda_i) \right)
= \exp\left( \sum_{i=1}^{m} \sum_{k=1}^\infty \frac{\lambda_i^k}{k} \right),
\end{align}
where the last equality comes from the Taylor expansion of each term $\ln(1-\lambda_i)$.  First, we will show that we can bound the higher powers of the eigenvalues.  By Lemma~\ref{lem:max_eigen_wishart}, for any $\beta \ge 4$, we have
\begin{align} \label{eq:lamba_max_bound}
\Pr[\lambda_{\max}(A) \ge 4\beta/\alpha^2] \le m e^{-2m(\beta-1)^2} \le \frac{\delta}{2},
\end{align}
whenever $m \ge 1/\beta^2 \ln(1/\delta)$.  Therefore, let's assume that $\lambda_{\max}(A) \le 4\beta/\alpha^2$.  We obtain
\begin{align}
\label{eq:taylor_remainder}
\sum_{i=1}^{m} \sum_{k=2}^\infty \frac{\lambda_i^k}{k} \le m \sum_{k=2}^\infty \frac{(4\beta/\alpha^2)^k}{k}
\le \frac{16 \beta^2 m}{\alpha^4},
\end{align}
where we have assumed that $8 \beta \le \alpha^2$ for the last inequality.  Combining \eqref{eq:log_expansion} and \eqref{eq:taylor_remainder}, we get
\begin{align}
\label{eq:N^2_formula_expanded}
 \prod_{i=1}^{m} (1- \lambda_i)^{-1}
\le \exp\left( \sum_{i=1}^{m} \lambda_i \right) \exp \left( \frac{16 \beta^2 m}{\alpha^4} \right)
= \exp( \Tr(A) )\exp \left( \frac{16 \beta^2 m}{\alpha^4} \right).
\end{align}
By Lemma~\ref{lem:moment_generating_wishart}, we have
\begin{align}
\mathbb E[\exp \Tr(A)] = \left(1-\frac{1}{\alpha^2 m}\right)^{-m^2} \le \exp(m/\alpha^2 + 1/\alpha^4)
\end{align}
where we've used a Taylor expansion and $\alpha \ge 2$ for the last inequality.  By Markov's inequality we get
\begin{align}
\label{eq:moment_generating_function_bound}
\exp \Tr(A) \le 2 \exp(m/\alpha^2 + 1/\alpha^4) / \delta,
\end{align}
with probability at least $1 - \delta/2$.  The theorem follows from \eqref{eq:N^2_formula_expanded}, \eqref{eq:moment_generating_function_bound}, and the union bound.
\end{proof}

\subsection{Proof of thermal state identity} \label{sec:thermal_state_identity}
The purpose of this section is to give a self-contained proof of the identity
\beq
\braket{n_i^2} = 2\braket{n_i}^2+\braket{n_i},
\eeq
which is used in \Cref{subapp:boundn} to show that $n$ is close to $m/\alpha^2$ with high probability. To start, recall that in the main text we introduced the two-mode squeezing operator between modes $i$ and $i+m$ with squeezing parameter $r$ as:
$S_2(r) = \exp\left[ r \left( a_i^\dagger a_{i+m}^\dagger - a_i a_{i+m} \right)\right]  = \exp\left[r G\right]$,
where $G = a_i^\dagger a_{i+m}^\dagger - a_i a_{i+m} =-G^\dagger$ is an anti-Hermitian generator. We would like to investigate expectation values of operators $O$ on the two-mode squeezed vacuum state
\begin{align}
\braket{O} = \braket{\psi|O|\psi} = \braket{0_i 0_{i+m}|\underbrace{S_2^\dagger(r) O S_2(r)}_{\equiv O(r)}|0_i 0_{i+m}} \text{ with }
\ket{\psi} = S_2(r) \ket{0_i 0_{i+m}}.
\end{align}
To make progress, we investigate the following quantity
\begin{align}
\frac{\partial}{\partial r}O(r) &= \left[\frac{\partial}{\partial r} e^{r G^\dagger } \right] O e^{r G } + e^{r G^\dagger}  O \left[\frac{\partial}{\partial r} e^{r G^\dagger}  \right] = \left[-G e^{r G^\dagger} \right] O e^{r G } + e^{r G^\dagger } O \left[G e^{r G} \right] =S_2(r)^\dagger [O,G]S_2(r).
\end{align}
In the last equation we used the fact that $G^\dagger = -G$ and introduced the commutator $[A,B] \equiv A B- BA$. Recalling the canonical bosonic commutation relations $[a_i,a_j]=[a_i^\dagger,a_j^\dagger]=0 $ and $[a_i,a_j^\dagger] = \delta_{i,j}$, we now set $O=a_i$ to obtain
\begin{align}
\begin{array}{r l}
\frac{\partial}{\partial r} a_i(r) &= a_{i+m}^\dagger(r), \\ \frac{\partial}{\partial r} a_{i+m}^\dagger(r) &= a_{i}(r),
\end{array} \Longrightarrow
\begin{array}{r l}
a_i(r) &= a_i \cosh r + a_{i+m}^\dagger \sinh r, \\  a_{i+m}^\dagger(r) &= a_{i+m}^\dagger \cosh r + a_{i} \sinh r.
\end{array}
\end{align}
Having derived these transformations we can now calculate
\begin{align}
\braket{\psi|a_i^{\dagger k} a_i^k|\psi}
&= \braket{0_i 0_{i+m}|S_2^\dagger(r)a_i^{\dagger k} S_2(r) S_2^\dagger(r) a_i^k S_2(r)|0_i 0_{i+m}} \nonumber \\
&=\braket{0_i 0_{i+m}| \left[  a_{i}^\dagger \cosh r + a_{i+m} \sinh r \right]^k \left[a_i \cosh r + a_{i+m}^\dagger \sinh r \right]^k|0_i 0_{i+m}} \nonumber \\
&=\braket{0_i 0_{i+m}| \left[  a_{i+m} \sinh r \right]^k \left[ a_{i+m}^\dagger \sinh r \right]^k |0_i 0_{i+m}} \nonumber \\
&= k! \sinh^{2 k} r.
\end{align}
In the last equation we used the canonical commutation relations and the fact that the vacuum is annihilated by the annihilation operator.
If we set $k=1$, we easily find the mean photon number to be $\braket{n_i} = \braket{a_i^\dagger a_i} = \sinh^2 r$. For the second moment, we write  $\braket{n_i^2} = \braket{a_i^\dagger a_i a_i^\dagger a_i} = \braket{a_i^{\dagger 2} a_i^2} + \braket{a_i^\dagger a_i}$. The first term in the RHS of the last equality can be obtained by setting $k=2$ and we finally obtain $\braket{n_i^2} = 2 \braket{n_i}^2 + \braket{n_i}$, as claimed.


\section{Technical lemmas for proving hardness of permanents with repetitions} \label{app:lemma_proofs}

In this appendix, we give proofs for the technical lemmas in \Cref{sec:collision_reduction}.  To start, let us show the proof of the expectation of a permanent of a matrix with repeated rows and columns which is used to define the allowable error fluctuations in the RGPE problem (see \Cref{def:RGPE}).

\repperexpectation*
\begin{proof}
For simplicity, let us define the $n \times n$ matrix $X := A_{S,T}$ so that we have $\Per(X) = \sum_{\sigma \in \mathrm S_n} \prod_{i =1}^n X_{i,\sigma(i)}$.  Because $X$ has repeated rows and columns, notice that many terms $\prod_{i=1}^n X_{i,\sigma(i)}$ of the permanent are the same for different permutations $\sigma$.  Let us rewrite the permanent so that such terms are grouped together.  Notice that $A_{S,T}$ can be thought of as consisting of $c^2$ blocks, where the $(i,j)$th block is the constant matrix $A_{i,j}^{s_i \times t_j}$.

We block the permutations similarly.  Concretely, given a permutation matrix over $n$ elements, let $f \colon \mathrm{S}_n \to \mathbb N^{c \times c}$ be the function that sums all elements within the same block.  For example, suppose we have $S = T = (2,2,1)$.  Then
$$
\left(\begin{array}{ c c | c c | c }
0 & 1 & 0 & 0 & 0 \\
0 & 0 & 0 & 0 & 1 \\ \hline
0 & 0 & 0 & 1 & 0 \\
0 & 0 & 1 & 0 & 0 \\ \hline
1 & 0 & 0 & 0 & 0
\end{array}\right)
\;\; \xrightarrow{\;\;\; f \;\;\;} \; \;
\left(\begin{array}{ c c c }
1 & 0 & 1 \\
0 & 2 & 0 \\
1 & 0 & 0
\end{array}\right).
$$
Let $\mathcal M_{S,T} := \{ f(\sigma) : \sigma \in \mathrm S_n \}$ be the set of matrices that arise from applying $f$ to all permutations.  We now ask, given a particular matrix $M \in \mathcal M_{S,T}$, how many times is $M$ the image of some permutation $\sigma$?  Notice that if we start from a permutation matrix and permute say the first $s_1$ rows or $t_1$ columns, then its image under $f$ is invariant.  Generalizing, we get
$$
| \{ \sigma \in \mathrm S_n : f(\sigma) = M \} | = \prod_{i = 1}^c  \prod_{j=1}^c  \frac{s_i! t_j!}{M_{i,j}!}
$$
where the denominator compensates for overcounting. Therefore, we can write
\begin{align*}
\Per(X) &= \sum_{\sigma \in \mathrm S_n} \prod_{i =1}^n X_{i,\sigma(i)} \\
&= \sum_{\sigma \in \mathrm S_n} \prod_{i=1}^c \prod_{j=1}^c A_{i,j}^{f(\sigma)_{i,j}} \\
&= \sum_{M \in \mathcal M_{S,T}} \prod_{i=1}^c \prod_{j=1}^c \frac{s_i! t_j!}{M_{i,j}!} A_{i,j}^{M_{i,j}}
\end{align*}
Recall that for a Gaussian random variable $a$, we have $\mathbb E[a^m (a^*)^{m'}]$ is $m!$ when $m = m'$, and 0 otherwise.  Therefore,
\begin{align*}
\mathbb E[|\Per(X)|^2] &= \mathbb E[\Per(X)\Per(X)^*] \\
&= \sum_{M \in \mathcal M_{S,T}} \prod_{i=1}^c \prod_{j=1}^c \left(\frac{s_i! t_j!}{M_{i,j}!}\right)^2 M_{i,j}! \\
&= n! \prod_{i=1}^c s_i! \prod_{j=1}^c t_j!
\end{align*}
where the last line comes from the fact
$$
n! = | \{ \sigma \in \mathrm S_n \} |
=  \sum_{M \in \mathcal M_{S,T}} | \{ \sigma \in \mathrm S_n : f(\sigma) = M \} |
=  \sum_{M \in \mathcal M_{S,T}} \prod_{i=1}^c  \prod_{j=1}^c \frac{s_i! t_j!}{M_{i,j}!}.
$$
\end{proof}

\subsection{Expressing permanent of matrix with repetitions as polynomial}

This section contains the proof of \Cref{thm:PerAsPoly}, which expresses the permanent of a matrix with repeated rows and columns as a polynomial whose constant coefficient is proportional to a permanent with no repetitions.  We start with the proof of correctness for the construction of the matrix $B'$ given in \Cref{sec:permanentreduc}:

\begin{lemma}\label{lemma:PerAsPoly}
Let $A \in \mathbb C^{c\times c}$ and $S \in \mathbb N^c$ such that $s_i \ge 1$ is given.  Defining $n':=c+k_S$, let $B'$ be a $c \times n'$ matrix obtained by adding $k_S$ columns to the matrix $A$ with entries that depend on complex variables $z$ and $y = \{y_{i,j}^{(\ell)}\}$.  Let $B_S'$ be the $n' \times n'$ matrix which has $s_i$ copies of row $i$ of $B'$.  We construct $B'$ such that
\begin{align}
\Per(B_S') &= \xi_1 \Per(A) + \gamma_1 z + \gamma_2 z^2 + \ldots + \gamma_{k_S} z^{k_S}, \\
\xi_1 &= \prod_{i=1}^c \Bigr(s_i! \prod_{j = 1}^{s_i - 1} y_{i,j}^{(i)}\Bigr)\label{eq:alpha1}
\end{align}
where the coefficients $\gamma_1,\ldots, \gamma_{k_S}$ are functions of the entries of $A$, $S$, and $y_{i,j}^{(\ell)}$. Furthermore, whenever $A \sim \mathcal N(0, 1)^{c \times c}_\mathbb C$,  $y_{i,j}^{(\ell)} \sim \mathcal N(0, 1)_\mathbb C$ are i.i.d., and $|z| = 1$, then we have that $B' \sim \mathcal N(0, 1)^{c \times n'}_\mathbb C$.
\end{lemma}
\begin{proof}
Recall that $B'$ is constructed as follows: for each $\ell$ such that $s_{\ell} > 1$, append to $A$ the $c\times (s_\ell-1)$ matrix $V^{(\ell)}$.  We have that $V^{(\ell)}_{ij} = y_{i,j}^{(\ell)}$ if $i = \ell$, and $z y_{i,j}^{(\ell)}$ otherwise.  We refer the reader to \Cref{sec:permanentreduc} for an explicit example of this construction.

To prove the correctness of this procedure, recall that the permanent of any matrix $C = \{c_{i,j}\}_{i,j=1}^n$ can be written as the sum of monomials $\prod_{i=1}^n c_{i, \sigma(i)}$ over all permutations $\sigma$. Let us consider the monomial terms in the expression for $\Per(B'_S)$ that are constant in $z$. We claim that for each such monomial, there must be at least one element from each row of the original matrix $A$. Suppose, by contradiction (and without loss of generality), that there is some monomial with no contribution from the elements in the first row of $A$.  By construction, it follows that the monomial must have $s_1$-many contributions from the first $s_1$ rows of the added columns (i.e., from $V^{(1)}, \ldots, V^{(c)}$).  Notice that only $V^{(1)}$ has entries in those first $s_1$ rows that are not multiplied by $z$.  However, $V^{(1)}$ only has $(s_1 -1)$-many columns, and hence the monomial in question must have some contribution from the the other added columns, which are all multiplied by the variable $z$.

Therefore, one element from each of the $c$ rows of $A$ appears in any monomial that is a constant in $z$.  Observe that these $c$ terms form a monomial of the permanent of $A$. The remaining $k_S$ terms must come from the matrices $V^{(i)}$, for which the only entries that do not depend on $z$ are of the form $y^{(i)}_{i, j}$.  Furthermore, since we cannot chose two terms from the same row, this immediately implies the only choice for the remaining terms is $\prod_{j = 1}^{s_i - 1} y_{i,j}^{(i)}$. To complete the proof, we simply need to count how many such terms occur.  For a particular monomial of the $A$, there are $s_i$ ways to choose the element $a_{i,\sigma(i)}$, and there are $(s_i - 1)!$ ways to choose the $\prod_{j = 1}^{s_i - 1} y_{i,j}^{(i)}$ term (since every permutation in the permanent of that submatrix suffices).  Eq.~\eqref{eq:alpha1} follows.  The remaining conclusions of the lemma are direct consequences of the construction, which we leave to the reader.
\end{proof}

Let us now describe in detail how \Cref{lemma:PerAsPoly} and its column version are combined to prove \Cref{thm:PerAsPoly}.  Given $A \in \mathbb C^{c\times c}$, we will construct a $(c+k_T) \times (c + k_S)$ matrix $B := B[z; x, y]$ to which we apply the repetition patterns $S = (s_1, \ldots, s_c, 1, \ldots, 1)$ and $T = (t_1, t_2, \ldots, t_c, 1, \ldots, 1)$.  There are $k_T$ trailing $1$'s in $S$, and $k_S$ trailing $1$'s in $T$. The matrix $B$ is obtained from $A$ by appending the rows/columns which arise in the following two invocations of the \Cref{lemma:PerAsPoly}:
\begin{enumerate}\setlength\itemsep{0pt}
    \item Row version with pattern $S' = (s_1, \ldots, s_c)$ to $A$ to obtain the $c \times n'$ matrix $B'$ in variables $z$ and $\{y_{i,j}^{(\ell)}\}$.
    \item Column version with pattern $T$ to the $n' \times n'$ matrix $B'_{S'}$ to obtain $k_T$ new rows in variables $z$ and $\{x_{i,j}^{(\ell)}\}$.
\end{enumerate}
Correctness of this procedure follows from the lemmas:
\begin{align}
    \Per(B_{S,T}) = \xi_2 \Per(B'_{S'}) + z p_2(z)
     = \xi_2 (\xi_1\Per(A) + z p_1(z)) + z p_2(z)
    = \xi \Per(A) + z p(z),
\end{align}
where $p(z)$ is a polynomial of degree $k-1$, and we have implicitly defined
\beq \label{eq:app_xi}
\xi := \xi_1 \xi_2 = \prod_{i=1}^cs_i!t_i! \left(\prod_{j=1}^{t_i -1}x^{(i)}_{i, j} \right) \left(\prod_{j=1}^{s_i -1}y^{(i)}_{i, j} \right).
\eeq

\subsection{Polynomial interpolation techniques}

Starting from Eq.~\eqref{eq:PerAsPoly}, which we restate,
\beq
\Per(B_{S,T}) =  \xi\Per(A) + \gamma_1 z + \gamma_2 z^2 + \ldots + \gamma_k z^k,
\eeq
we give the proofs for the polynomial interpolation theorems given in \Cref{sec:poly_interpolation} to extract the constant coefficient.  We begin with a general technique for approximating the constant term in a complex polynomial:

\begin{lemma}\label{lemma:poly_interpolation1}
Let $P(z) = \gamma_0 + \sum_{n=1}^k \gamma_n z^n$ be a complex polynomial. Assume that there is a procedure to obtain an estimate $\tilde{P}(z)=P(z) + e(z)$ for any $z$, where $e(z)$ is the error in the estimate, satisfying $|e(z)|\leq \epsilon$ for all $z$. Then
\beq
\tilde{\gamma}_0 = \frac{1}{k+1}\sum_{j=0}^k \tilde{P}\left(e^{2\pi i j/(k+1)}\right)
\eeq
is an estimate of $\gamma_0$ such that $|\tilde{\gamma}_0 - \gamma_0|\leq \epsilon$.
\end{lemma}
\begin{proof}
The proof is by construction. For each $j=0, 1, \ldots, k$, define $z_j=e^{2\pi i j/(k+1)}$. Then it holds that
\begin{align*}
    \sum_{j=0}^k P(z_j) &= \sum_{j=0}^k \left( \gamma_0 + \sum_{n=1}^k \gamma_n z_j^n \right)\\
    &= (k+1)\gamma_0 + \sum_{n=1}^k \gamma_n \sum_{j=0}^k z_j^n\\
    &= (k+1)\gamma_0,
\end{align*}
where we used the fact that $\sum_{j=0}^k z_j^n=0$ for all $n\leq k$. It follows that
\beq
\gamma_0 = \frac{1}{k+1}\sum_{j=0}^k P(z_j).
\eeq
We then have
\begin{align*}
    |\tilde{\gamma}_0 - \gamma_0| &= \Bigr|\frac{1}{k+1}\sum_{j=0}^k \tilde{P}(z_j)- P(z_j)\Bigr|\\
    &= \frac{1}{k+1}\Bigr|\sum_{j=0}^k e(z_j)\Bigr|\leq \epsilon,
\end{align*}
as desired.
\end{proof}

Combining lemmas \ref{lemma:PerAsPoly} and \ref{lemma:poly_interpolation1}, we have shown that given an oracle to estimate $\Per(B_{S,T})$, it is possible to also estimate $\Per(A)$ by calling the oracle $k+1$ times for values of $z$ equally distributed along the unit circle.
Indeed, defining $\tilde{P}=\tilde{\gamma_0}/\xi$, with $\xi$ as in Eq.~\eqref{eq:xi_def}, we have $|\tilde{P}-\Per(A)|$ is at most the error of the oracle times $1/|\xi|$. We now use this to provide a reduction from $\GPEa$ to $\RGPEa$, as per \Cref{def:RGPE}.

\GPEaToRGPEa*
\begin{proof}
Let $A \in \mathbb C^{c \times c}$ be the input matrix to $\GPEa$. Following \Cref{thm:PerAsPoly}, we build a matrix $B$ such that $\Per(B_{S,T})=\xi \Per(A) + z p(z)$. Recall that if we chose $|z|=1$ and the variables $x^{(\ell)}_{i, j}, y^{(\ell)}_{i, j}$ to be from the same Gaussian distribution as the entries of $A$, this ensures that $B_{S,T}$ is a Gaussian random matrix (with repetitions) that can be fed into the oracle for $\RGPEa$.

Now repeat this construction for $z_j=e^{2\pi i j/(k+1)}$ to obtain the matrices $B[z_j; x, y]_{S,T}$ for each $j=0,1,\ldots, k+1$.
For each such matrix, use the oracle to compute an estimate $\tilde{P}(z_j)$ of its permanent. Then we output an estimate $\tilde{P}$ of $\xi \Per(A)$ given by
\beq
\tilde{P} = \frac{1}{(k+1)}\sum_{j=0}^k \tilde{P}(z_j).
\eeq
Using \Cref{lemma:poly_interpolation1}, we have $|\tilde{P} - \xi \Per(A)| \le \varepsilon \sqrt{(c+k)! \prod s_i! t_j!}$, and so the error in the estimate $\tilde{P}/\xi$ for $\Per(A)$ is
\begin{align*}
    |\tilde{P}/\xi - \Per(A)| = \frac{1}{|\xi|}|\tilde{P} - \xi\Per(A)| \leq \left(\varepsilon\frac{1}{|\xi|}\sqrt{((c+k)!/c!) \prod s_i! t_j!}\right)\sqrt{c!} = \epsilon \sqrt{c!}
\end{align*}
as claimed.
\end{proof}

To understand $\epsilon$ in \Cref{thm:GPEa_to_RGPEa}, we must understand the size of $|\xi|$.  To this end, in \Cref{sec:xi_bound} we will prove that $|\xi| = \Omega(1.498^k)$ with constant probability for sufficiently large $k$.  To succeed with probability $1-\delta$ for any $\delta > 0$, we note that the random Gaussian variables appearing in the definition of $|\xi|$ come from the additional Gaussian variables (the $x^{(\ell)}_{i,j}$ and $y^{(\ell)}_{i,j}$ terms) employed in the reduction.  Therefore, if those variables are not sufficiently large, we can simply resample them $O(\log(1/\delta))$-many times until they are.

Let us now turn to our second polynomial interpolation technique, which works directly with the magnitude of the permanent:

\GPEpaToRGPEpa*
\begin{proof}
We begin by taking the absolute squared value of \Cref{eq:PerAsPoly}. This leads to an equation of the form
\beq
\left|\Per(B_{S,T})\right|^2
    = \left|\xi\right|^2 \left|\Per(A)\right|^2 + q(z) \label{eq:perzxy2}
\eeq
where $q(z)$ is polynomial in $z$ and $\bar{z}$ with $q(0)=0$. Contrary to \Cref{thm:GPEa_to_RGPEa}, we cannot compute the left-hand side for $z$ in the unit circle and leverage the corresponding cancellations, since $q(z)$ has monomials that depend on $\left|z\right|^2$. However, notice that when $z$ is real, \Cref{eq:perzxy2} is of the same form as Eq.\ (8) in \cite{AaronsonBrod}.  Unsurprisingly, the remainder of this proof closely follows that in \cite{AaronsonBrod}, and so, we will just walk through their proof outline and emphasize the differences that arise in our case but omit repetition of technical details.

The first step is to choose $x^{(i)}_{i, j}$ and $y^{(i)}_{i, j}$ in $B_{S,T}$ to be i.i.d.\ Gaussian values from $\mathcal N(0,1)_\mathbb{C}$ (i.e., the same distribution as the entries of $A$), and then choose different (real and positive) values of $z$ in the range $[1-\gamma, 1+\gamma]$, for a suitably small $\gamma$. Note that, if $\gamma = 0$, then matrix $B$ has only i.i.d.\ elements from $\mathcal N(0,1)_\mathbb{C}$, as per \Cref{thm:PerAsPoly}. When we choose a different value of $\gamma$ we induce a deformation in the distribution of $B$, since now some of its elements are drawn from $\mathcal N(0,z^2)_\mathbb{C}$ instead.

If $\gamma$ is too large, then the oracle for $\RGPEpa$ may fail since it only succeeds with high probability over the Gaussian distribution.  To bound how large $\gamma$ can be before this happen, we compute the total variation distance between the distributions over $B$ when $\gamma=0$ and when $\gamma > 0$ (its deformed version). If that distance is $O(\delta)$, then the oracle for $\RGPEpa$ will still succeed with probability $O(\delta)$ and can be invoked to compute $\left|\Per(B_{S,T})\right|^2$ to the desired precision. Luckily, there is a closed form for the Kullback-Leibler divergence between two sets of i.i.d.\ Gaussian variables which, combined with Pinsker's inequality, gives a bound on the desired total variation distance (cf.\ Lemma 1 of \cite{AaronsonBrod}). The only thing remaining is to count how many matrix elements of $B$ have actually been deformed via multiplication by $z$. Following the proof of \Cref{thm:PerAsPoly}, we can see that $(c^2+k_S+k_T)$ elements of $B$ are standard complex Gaussians, and the remaining $(c-1)k+k_S k_T$ have been deformed. This implies that we must set
\beq
\gamma \leq \frac{\delta}{\sqrt{(c-1)k+k_S k_T}},
\eeq
for the two distributions to be sufficiently close.

The second step is to apply the polynomial interpolation. This is done by invoking the oracle to estimate the left-hand side of \Cref{eq:perzxy2}, to error $\varepsilon (c+k)! \prod s_i! t_j!$, for $2k+1$ values of $c$ evenly spaced in the interval $[1-\gamma, 1+\gamma]$. This leads to a linear system which can be solved by a standard least squares method. If we denote by $\hat\beta_0$ the resulting estimator for the constant term in the polynomial of \Cref{eq:perzxy2}, the variance of this estimator is given by
\beq
\Var[\hat\beta_0] = \frac{\left(\varepsilon (c+k)! \prod s_i! t_j!\right)^2}{\gamma^{4k}},
\eeq
which is the same as Eq.\ (20) in \cite{AaronsonBrod} when translated to our variables. By Chebyshev's inequality, we get
\beq
\Pr\left[\left| |\Per(A)|^2 - \hat\beta_0 |\xi|^{-2} \right| \ge |\xi|^{-2} \sqrt{\Var[\hat\beta_0]/\delta} \right] \le \delta,
\eeq
and so we obtain a bound for $\left|\Per(A)\right|^2$ that holds with probability $1-\delta$. We identify this bound with the error in our estimate for $\left|\Per(A)\right|^2$, which needs to be at most $\epsilon c!$ by the definition of $\GPEpa$. Combining all of the above, we obtain the following relation
\beq \label{eq:partialerrorbound}
\epsilon
= \varepsilon \frac{(c+k)! \left((c-1)k + k_S k_T \right)^k}{c! X \delta^{2k + 1/2} \prod s_i! t_j!}
\le \varepsilon \frac{(c+k)^{2k} k^k}{X \delta^{2k+1/2}2^k},
\eeq
where $X = \prod_{j=1}^{t_i -1}|x^{(i)}_{i, j}|^2 \prod_{j=1}^{s_i -1}|y^{(i)}_{i, j}|^2$ is the product of $k$ independent absolute-squared standard i.i.d.\ complex Gaussian variables (coming from the definition of $\xi$ in \Cref{eq:app_xi}). Each of the $k$ terms, which we will generically call $x$, follows an exponential distribution with CDF
\beq
\Pr[x\leq g] = 1 - e^{-g}.
\eeq
If we set $g = \delta /k$, then $\Pr[x\leq g] \leq 1- (1-\delta)^{1/k}$, or in other words $\Pr[x > g] > (1-\delta)^{1/k}$.  By the independence of the $k$ terms comprising $X$, we get $\Pr[1/X < 1/g^k] > 1-\delta$.
Therefore we can assume that, with probability at least $1-\delta$, $1/X$ is at most $k^k/\delta^k$. Plugging this into \Cref{eq:partialerrorbound}, we get that
\beq
\epsilon = O\biggl(\frac{(c+k)^{2k} k^{2k}}{\delta^{3k+1/2}2^k} \varepsilon \biggr),
\eeq
which is polynomial in $c$, $\delta$ and $\varepsilon$ as long as $k=O(1)$, as desired.
\end{proof}

\subsection{Lower bound for \texorpdfstring{$|\xi|$}{reduction coefficient}} \label{sec:xi_bound}

The purpose of this section is to give a lower bound on the $|\xi|$ term coming from \Cref{eq:xi_def} in \Cref{thm:GPEa_to_RGPEa}:
\beq
|\xi| = \prod_{i=1}^cs_i!t_i! \left(\prod_{j=1}^{t_i -1}|x^{(i)}_{i, j}| \right) \left(\prod_{j=1}^{s_i -1}|y^{(i)}_{i, j}| \right). \label{eq:xi_abs}
\eeq
where $x^{(i)}_{i, j}, y^{(i)}_{i, j}$ are i.i.d.\ standard complex Gaussians. We show
\begin{theorem}\label{thm:xi_bound}
$\Pr[|\xi| \ge 1.498^k] \ge 1/2 - 4/\sqrt{k}$.
\end{theorem}
To prove this theorem, let us first focus on bounding the product of the $k$ Gaussian magnitudes that appear in Eq.~\eqref{eq:xi_abs}. To simplify the notation, let us simply rewrite this product as $X = \prod_{i = 1}^k X_i$ where each $X_i$ is the random variable for the magnitude of a standard complex Gaussian.  One can check that $X_i \sim \textrm{Rayleigh}(1/2)$ such that $\Pr[X_i = x] =2x e^{-x^2}$ for $x \in [0,\infty)$. To bound $X$, we will need the following standard central limit theorem:
\begin{theorem}[Berry-Esseen Theorem] \label{thm:berry_esseen}
Let $Z_1,\ldots,Z_k$ be real i.i.d. random variables satisfying
\begin{align*}
   \mathbb{E}[Z_i] &= \mu,\\
   \mathbb{E}[(Z_i-\mu)^2] &= \sigma^2>0,\\
   \mathbb{E}[|Z_i-\mu|^3] &= \rho<\infty
\end{align*}
Let $Z = \sum_{i=1}^k Z_i$ and let $W \sim \mathcal{N}(\mu k,\sigma^2k)_\mathbb{R}$ be a real Gaussian. Then, for all $x \in \mathbb{R}$,
\beq
|\Pr[Z > x] - \Pr[W > x]| \leq \frac{C\rho}{\sigma^3\sqrt{k}},
\eeq
where $C$ is some universal constant.
\end{theorem}

\begin{lemma}
\label{lem:X_lower_bound}
$\Pr[X \ge .7493^k] \ge 1/2 - 4/\sqrt{k}$.
\end{lemma}
\begin{proof}
Letting $Z_i = \log(X_i)$ and $Z = \sum_{i=1}^k Z_i$, we begin by writing
\beq
X = e^{\sum_{i=1}^k \log(X_i)} = e^Z.
\eeq
One can check that $\Pr[Z_i = z] = 2e^{2z - e^{2z}}$ leading to the following moments:
\begin{align}
   \mathbb{E}[Z_i] &= -\gamma/2, \label{eq:Z_expectation} \\
   \mathbb{E}[(Z_i+\gamma/2)^2] &= \pi^2/24, \label{eq:Z_variance} \\
   \mathbb{E}[|Z_i+\gamma/2|^3] &< 0.5136 \label{eq:Z_third_moment}
\end{align}
where $\gamma \approx 0.5772$ is the Euler-Macheroni constant.
Plugging Eqs.~\eqref{eq:Z_expectation}, \eqref{eq:Z_variance}, and \eqref{eq:Z_third_moment} into the Berry-Esseen Theorem (using $C < .4748$ \cite{shevtsova2011absolute}), we get that for any $x \in \mathbb{R}$
\beq \label{eq:berry_essen_plugged_in}
|\Pr[Z > x] - \Pr[W > x]| < 4/\sqrt{k}
\eeq
where $W$ is real Gaussian with mean $-k\gamma/2$ and variance $k\pi/24$.  This immediately gives the lower bound
\begin{align*}
\Pr[Z > -k\gamma/2] &= \Pr[Z > -k\gamma/2] - \Pr[W > -k\gamma/2] + \Pr[W > -k\gamma/2] \\
&\ge \Pr[W > -k\gamma/2] - |\Pr[Z > -k\gamma/2] - \Pr[W > -k\gamma/2]| \\
&\ge 1/2 - 4/\sqrt{k}
\end{align*}
where the last line uses Eq.~\eqref{eq:berry_essen_plugged_in} and the simple fact that a Gaussian is greater than its mean with probability $1/2$.  Of course, if $Z > -k\gamma/2$, then $X > (e^{-\gamma/2})^k$, which yields the lemma.
\end{proof}

Using that $\prod_{i=1}^c s_i!t_i! \ge 2^k$ immediately gives \Cref{thm:xi_bound}.


\section{Moments of the click distribution}\label{app:clicks}
In this appendix we derive the results provided in \Cref{sec:collisions} on click moments. In \Cref{app:covmats}, we present a brief review of the Gaussian formalism in terms of covariance matrices an apply it to derive the Husimi covariance matrix of a Gaussian state in which a transition matrix has been encoded.
In \Cref{Sec:collisions}, we describe how to reduce the calculation of the mean number of clicks and its variance for an arbitrary multimode state to the calculation of one- and two-mode vacuum probabilities.  In \Cref{sec:fixing_scaling}, we show how to fix the scaling factor $\alpha$ in terms of the photon number density using the quarter circle law. Finally, in \Cref{sec:coup}, we put together all the results derived previously and derive Eqs.~\eqref{app:means} and~\eqref{equations} giving the first and second order moments of the click distribution for the two halves of the modes respectively.

\subsection{Gaussian states and covariance matrices}\label{app:covmats}
In this appendix we provide a self-contained review of the Gaussian formalism as relevant to the discussion in the main text and the subsequent appendices. We consider a Gaussian state $\rho$ with $M$ modes and zero mean. The latter condition means that $\braket{a_i} = \text{Tr}(a_i \rho)= 0$ for all $i \in \{1,\ldots, M\}$.
The state $\rho$ is uniquely characterized by its (complex-)Husimi covariance matrix ${Q}$ with entries
\begin{align}
Q_{i,j} = \tfrac{1}{2} \braket{ z_i z_j^\dagger + z_j^\dagger z_i} +\tfrac{1}{2} \delta_{i,j} = \tfrac{1}{2} \text{Tr}\left( \rho \left[ z_i z_j^\dagger + z_j^\dagger z_i\right]\right)+\tfrac{1}{2} \delta_{i,j},
\end{align}
where
\begin{align}
{z} = \left(a_1,\ldots,a_M, a_1^\dagger,\ldots,a_M^\dagger \right).
\end{align}
For $Q$ to represent a valid quantum mechanical state it needs to satisfy the uncertainty relation~\cite{FR-serafini2017}
\begin{align}
Q + \tfrac12 Z -\tfrac12 \mathbb{I}_{2M} \succeq 0,
\end{align}
where $Z = \left(\begin{smallmatrix}
\mathbb{I}_M & 0 \\
0 & -\mathbb{I}_M
\end{smallmatrix} \right)$, with $\mathbb{I}_M$ the $M \times M$ identity matrix.

The marginal state of a subset of the modes of a Gaussian state is another Gaussian state. In particular the covariance matrix of a subset of the modes $i_1,\ldots,i_K$ is obtained by keeping rows and columns $i_1,\ldots,i_K$ and $i_1+M,\ldots,i_K+M$ from the original covariance matrix $Q$.
Finally, note that the diagonal elements of the Husimi covariance matrix are related to the mean photon number in each of the modes as
\begin{align}
Q_{i,i} = Q_{i+M, i+M} = \braket{a^\dagger_i a_i} +1.
\end{align}
For a given zero-mean Gaussian state the probability of obtaining vacuum in all its modes when measuring its photon number statistics is given by
\begin{align}
\Pr(\text{vac}) = \frac{1}{\sqrt{\text{det}(Q)}}.
\end{align}
A simple expression exists linking a pure Gaussian state parametrized by an adjacency matrix $A=A^T \in \mathbb{C}^{M \times M}$ and its Husimi covariance matrix~\cite{hamilton2017,jahangiri2020point}
\begin{align}
Q = \left( \mathbb{I}_{2M} - \mathcal{X} (A \oplus A^*) \right)^{-1}, \text{ where  }\mathcal{X} = \left(\begin{smallmatrix}
0 & \mathbb{I}_M   \\
\mathbb{I}_M & 0
\end{smallmatrix} \right).
\end{align}

We now consider the case where $M = 2m$ and we use $A$ to encode a transition matrix $C \in \mathbb{C}^{m \times m}$ as follows
\begin{align}
A = \begin{pmatrix}
0 & C \\
C^T & 0
\end{pmatrix}.
\end{align}
We write the following decomposition $C = U \sigma V^T$ with $U$ and $V$ $m\times m$ unitary matrices and $\sigma$ a diagonal matrix with $0 \leq \sigma_{i,i} <1$.
We can then write
\begin{align}
\mathcal{X} (A \oplus A^*) &= \begin{pmatrix}
0 & 0 & 0 & C \\
0 & 0 & C^T & 0 \\
0 & C^* & 0 & 0\\
C^\dagger & 0 & 0 & 0
\end{pmatrix} =
F \left(
\begin{array}{cccc}
\sigma & 0 & 0 & 0 \\
0 & \sigma & 0 & 0 \\
0 & 0 & -\sigma & 0 \\
0 & 0 & 0 & -\sigma \\
\end{array}
\right)
{F}^\dagger ,
\end{align}
where
\begin{align}
{F} = \frac{1}{\sqrt{2}}\left(
\begin{array}{cccc}
U & 0 & -U & 0 \\
0 & V & 0 & -V \\
0 & U^* & 0 & U^* \\
V^* & 0 & V^* & 0 \\
\end{array}
\right),
\end{align}
is a unitary matrix of size $4m \times 4m$.
With this diagonalization it is direct to write
\begin{align}
Q = \begin{pmatrix}
X & 0 & 0 & W \\
0 & Y & W^T & 0 \\
0 & W^* & X^* & 0 \\
W^\dagger & 0 & 0 & Y^*
\end{pmatrix},
\end{align}
with
\begin{align}\label{eq:Qblocks}
X = U \frac{\mathbb{I}_m}{\mathbb{I}_m -\sigma^2}  U^\dagger = X^\dagger, \quad Y = V \frac{\mathbb{I}_m}{\mathbb{I}_m -\sigma^2}  V^\dagger = Y^\dagger, \quad W = U \frac{\sigma}{\mathbb{I}_m -\sigma^2}  V^T.
\end{align}
This form of the covariance matrix will be useful below.

\subsection{Detector click statistics}\label{Sec:collisions}
In this section we derive simple expressions for the mean and variance of the number of clicks when a quantum state is probed using threshold detectors.
We will obtain expressions for $\braket{c}$, $\braket{c^2}$ and $\Delta^2 c$ that hold for an arbitrary $M$-mode quantum state. In the next subsections we will specialize these expressions to zero-mean Gaussian states encoding a random complex-Gaussian transition matrix.

Let $\nu_c$ be a subset of $\{1,2,\dots,M\}$ where $\nu_c$ contains exactly $c$ elements. Furthermore, let $p(\nu_c)$ denote the probability of observing a click in all detectors with labels in $\nu_c$. This probability is given by~\cite{quesada2018gaussian}
\begin{align}
    p(\nu_c) = \text{Tr}\left[\rho\left(\prod_{j\in \nu_c} P_1^{(j)} \prod_{j\notin \nu_c} P_0^{(j)} \right)\right],
\end{align}
where $P_0^{(j)} = |0_j\rangle\langle 0_j|$ is a projector onto the vacuum in mode $j$ and $P_1^{(j)}=\id^{(j)}-P_0^{(j)}$.
The probability of observing $c$ clicks is then given by the sum $p(c) = \sum_{\nu_c} p(\nu_c)$.
Now consider its characteristic function:

\begin{align}
    \chi_c(x) &= \braket{e^{i c x}}=\sum_{c=0}^{M} e^{i x c} p(c) \cr &=
    \sum_{c=0}^{M} e^{ixc}\sum_{\nu_c} \text{Tr}\left[\rho\left(\prod_{k\in \nu_c}  P_1^{(k)}
        \prod_{k\notin \nu_c} P_0^{(k)}\right) \right]
        \cr &= \text{Tr}\left[\rho\sum_{c=0}^{M}\sum_{\nu_c}\left(\prod_{k\in \nu_c}  e^{ix}P_1^{(k)}
        \prod_{k\notin \nu_c} P_0^{(k)}\right) \right]
        \cr &=
        \text{Tr}\left[\rho\prod_{\ell=1}^{M} \left(e^{ix} P_1^{(\ell)} + P_0^{(\ell)}\right)\right]~,
\end{align}
where in the last line we used the binomial expansion. The expected number of detectors that click $\braket{c}$ can then be computed from the characteristic function \cite{banchi2020molecular}
\begin{align}
\braket{c} &= -i \frac{d \chi_c(x)}{dx}\Bigr|_{x=0}
\cr &=-i\text{Tr}\left.\left[\rho\frac{d}{dx}\prod_{\ell=1}^{M} \left(e^{ix} P_1^{(\ell)} + P_0^{(\ell)}\right)\right]\right|_{x=0}
\cr &= \text{Tr}\left[\rho \sum_{j=1}^{M}P_1^{(j)}\prod_{\ell\neq j}( P_1^{(\ell)} + P_0^{(\ell)})\right]
\cr &= \sum_{j=1}^{M} \text{Tr}\left[P_1^{(j)} \rho\right],
\end{align}
where we have used the product rule for derivatives. Using the fact that $P_1^{(j)} + P_0^{(j)}=\id^{(j)}$, we conclude that
\beq
\braket{c} = M - \sum_{j=1}^{M}\langle P_0^{(j)}\rangle,
\eeq
where we have defined $\langle P_0^{(j)}\rangle:= \text{Tr}\left[P_0^{(j)} \rho\right]$. Similarly, we can obtain the second moment of the number of clicks as
\beq
\braket{c^2} = \left. (-i)^2 \frac{d^2 \chi_c(x)}{dx^2} \right|_{x=0}.
\eeq
We now write
\begin{align}
\frac{d^2\chi_c(x)}{dx^2} &= \sum_{j=1}^{M} \text{Tr} \left( \rho
\frac{d}{dx}
\prod_{\ell\neq j } \left[e^{i x}P_1^{(\ell)}+P_0^{(\ell)} \right] i P_1^{(j)} e^{i x} \right)\cr
&= \sum_{j=1}^{M} \text{Tr} \left( \rho
 \Bigr\{i^2  P_1^{(j)} e^{i x} \prod_{\ell \neq j} \left[e^{i x} P_1^{(j)} + P_0^{(\ell)} \right] +  i^2 P_1^{(j)} e^{i x} \sum_{k\neq j }^{M} P_1^{(k)} e^{i x} \prod_{\ell \neq j, \ell\neq k} \left[e^{i x} P_1^{(\ell)}+ P_0^{(\ell)} \right]\Bigr\}\right)\nonumber.
\end{align}
By setting $x=0$ in the last equation we obtain
\beq
\braket{c^2} = \sum_{j=1}^{M} \text{Tr}[P_1^{(j)}\rho_j] + \sum_{j=1}^{M}\sum_{j\neq k}^{M} \text{Tr}[P_1^{(j)} P_1^{(k)}\rho].
\eeq
Equivalently, we can use $P_1^{(j)} = \id^{(j)} - P_0^{(j)}$, which gives
\beq
\braket{c^2} = M^2 - (2 M -1) \sum_{j=1}^{M} \braket{P_0^{(j)}} + 2 \sum_{j<k}^{M} \braket{P_0^{(j)} P_0^{(k)}},
\eeq
where we have defined $\braket{P_0^{(j)} P_0^{(k)}}:= \text{Tr}\left[P_0^{(j)} P_0^{(k)} \rho\right]$.
Putting these results together, the variance of the number of clicks is
\beq
\label{eq:variance}
\Delta^2 c = \sum_{j=1}^{M} \braket{P_0^{(j)}} (1-\braket{P_0^{(j)}})+ 2 \sum_{j<k} \left( \braket{P_0^{(j)} P_0^{(k)}} - \braket{P_0^{(j)}} \braket{P_0^{(k)}} \right) = \sum_{i,j=1}^{M} \left( \braket{P_0^{(j)} P_0^{(k)}} - \braket{P_0^{(j)}} \braket{P_0^{(k)}} \right).
\eeq

Since reduced states of Gaussian states can be computed efficiently~\cite{quesada2020exact}, it is possible to use these formulas to calculate the average and variance of the number of clicks for these states. Under the assumption of zero-mean Gaussian states, the two types of terms in the equations above can be written as
\begin{align}
\braket{P_0^{(j)} P_0^{(k)}}  = \frac{1}{\sqrt{\det(Q^{i,j})}}, \quad \braket{P_0^{(j)} } = \frac{1}{\sqrt{\det(Q^{j})}},
\end{align}
where $Q^{i,j}$ and $Q^{j}$ are the reduced Husimi covariance matrices of modes $i,j$ and $j$ respectively. We note that these quantities can also be estimated efficiently for Gaussian states using Monte Carlo methods as recently shown in Ref.~\cite{drummond2021simulating}.

In the following, we use these results to estimate the average number of clicks for GBS devices that are configured to encode Gaussian random matrices.

\subsection{Fixing the scaling factor in terms of the photon number density}\label{sec:fixing_scaling}
In Eq.~\eqref{mean_scaling} we used the quarter circle law to find a relation between the scaling factor $\alpha$, the number of modes $m$ and the total mean photon number. From this equation we can solve
\eq{
\alpha(\mu) = \frac{1+\mu}{\sqrt{\mu}}.
}
For future convenience we can also calculate the following matrix norms of the blocks of the Husimi covariance matrix (recall Eq.~\eqref{eq:Qblocks}) when the transition matrix has singular values satisfying the quarter circle law
\eq{
\mathbb{E}\left( ||X||^2 \right) &= \mathbb{E}\left( ||Y||^2 \right)  = \mathbb{E}\left( \sum_{i=1}^m \left(\frac{1}{1-\sigma_i^2} \right)^2 \right) = m \int_{0}^{2/\alpha(\mu)} p_{\alpha(\mu)}(\sigma) \left(\frac{1}{1-\sigma^2} \right)^2 = m \frac{ (1+\mu)}{1-\mu},\\
\mathbb{E}\left(||W||^2\right) &= \mathbb{E}\left(\sum_{i=1}^m \left(\frac{\sigma_i}{1-\sigma_i^2} \right)^2 \right) = m \int_{0}^{2/\alpha(\mu)} p_{\alpha(\mu)}(\sigma) \left(\frac{\sigma}{1-\sigma^2} \right)^2 = m \frac{\mu (1+\mu)}{1-\mu}.
}
In the equation above we used the fact that the
Frobenius norm of a matrix, $||A||^2 = \sum_{i,j=1}^m |A_{i,j}|^2$ can also be expressed as the sum of the squares of the singular values of the same matrix.

\subsection{Means and variances of the click distribution}~\label{sec:coup}
With the preliminary calculations in the last subsections we are ready to tackle the calculation of the means and covariances of the click distribution between the two sets of modes when averaged over the set of Gaussian random matrices.

For this we start with the means, for which we find the covariance of mode $i$ in the first half
\eq{
Q^{i} = \begin{pmatrix}
    X_{i,i} & 0\\
    0 &X_{i,i}
\end{pmatrix}    \longrightarrow  \braket{P_0^{(i)}} = \frac{1}{X_{i,i}} = \frac{1}{1+\braket{a_i^\dagger a_i}},
}
and similarly if $i=j+m$ is in the second half
\eq{
Q^{i} = \begin{pmatrix}
        Y_{j,j} & 0\\
        0 &Y_{j,j}
    \end{pmatrix}    \longrightarrow  \braket{P_0^{(i)}} = \frac{1}{Y_{j,j}} = \frac{1}{1+\braket{a_i^\dagger a_i}}.
}
From the last two equations we find $\braket{d} = m - \sum_{i=1}^m \frac{1}{X_{j,j}}$ and similarly $\braket{e} = m - \sum_{i=1}^m \frac{1}{Y_{j,j}}$. Now we can let $X_{j,j} \to 1+\mu$ and $Y_{j,j} \to 1+\mu$ and rearrange to find precisely Eq.~\eqref{app:means}.

Now we consider the covariances between the clicks in the two halves. For this we need the two-body covariance matrices of the modes. Consider first the case where $i$ and $j$ refer to modes in the first half of the set, i.e., $i,j\leq m$. In that case their reduced covariance matrix is given by
\eq{\label{eq:qij1half}
    Q^{i,j} = \begin{pmatrix}
        X_{i,i} & X_{i,j} & 0 & 0\\
        X_{i,j}^* & X_{j,j} & 0 & 0\\
        0 & 0 & X_{i,i} & X_{i,j}^* \\
        0 & 0 & X_{i,j} & X_{j,j}
    \end{pmatrix} \longrightarrow \braket{P_0^{(j)} P_0^{(i)}} - \braket{P_0^{(j)}} \braket{P_0^{(j)}} &= \frac{1}{X_{i,i} X_{j,j}-|X_{i,j}|^2} - \frac{1}{X_{i,i} X_{j,j}} \\
    &=\frac{1}{X_{i,i} X_{j,j} }\sum_{k=1}^{\infty} \left[ \frac{|X_{i,j}|^2}{X_{i,i}X_{j,j}} \right]^k.
}
Note that the Taylor expansion is guaranteed to converge since the uncertainty relation for the $Q$ covariance matrix (cf. Appendix \ref{app:covmats}) guarantees that $X_{i,i}X_{j,j}>|X_{i,j}|^2$.
If the two modes $i,j$ are in the second half one can obtain a similar expression to the one above by letting $X \to Y$

Finally, if $i$ and $j$ are in different halves one has
\eq{\label{eq:qijdiffhalf}
    Q^{i,j} = \begin{pmatrix}
        X_{i,i} & 0 & 0 & W_{i,j} \\
        0 & Y_{j,j} & W_{i,j} & 0\\
        0 & W_{i,j}^* & X_{i,i} & 0\\
        W_{i,j}^* & 0 & 0 & Y_{j,j}
    \end{pmatrix} \longrightarrow \braket{P_0^{(j)} P_0^{(i)}} - \braket{P_0^{(j)}} \braket{P_0^{(j)}} &= \frac{1}{X_{i,i} Y_{j,j}-|W_{i,j}|^2} - \frac{1}{X_{i,i} Y_{j,j}} \\
    &=\frac{1}{X_{i,i} Y_{j,j} }\sum_{k=1}^{\infty} \left[ \frac{|W_{i,j}|^2}{X_{i,i}Y_{j,j}} \right]^k.
}
As before, the Taylor expansion converges by virtue of the uncertainty relation.

We can now look at the variance of the total number of clicks (recall Eq.~\eqref{eq:variance}) over the two halves to write it as
\eq{
    \Delta^2[d+e] =& \sum_{i=1}^m \left(\frac{1}{X_{i,i}}+\frac{1}{Y_{i,i}}-\frac{1}{X_{i,i}^2} -\frac{1}{Y_{i,i}^2} \right) \\
    &+  \sum_{i \neq j}^m \left(  \frac{1}{X_{i,i} X_{j,j} - |X_{i,j}|^2} - \frac{1}{X_{i,i} X_{j,j}} + \frac{1}{Y_{i,i} Y_{j,j} - |Y_{i,j}|^2} - \frac{1}{Y_{i,i} Y_{j,j}} \right) \nonumber\\
    &+ 2\sum_{i, j=1}^m \left(  \frac{1}{X_{i,i} Y_{j,j} - |W_{i,j}|^2} - \frac{1}{X_{i,i} Y_{j,j}} \right). \nonumber
}
We will now subtract the quantity $\sum_{i=1}^m \frac{1}{X_{i,i}^2}+\frac{1}{Y_{i,i}^2}$ in the first term of the equation above and add it in the second. Then we will also Taylor expand the terms in the second and third line to write
\eq{
    \Delta^2 [d+e] =& \sum_{i=1}^m \left(\frac{1}{X_{i,i}}+\frac{1}{Y_{i,i}}-\frac{2}{X_{i,i}^2} -\frac{2}{Y_{i,i}^2} \right) \\
    &+  \sum_{i ,j}^m \frac{|X_{i,j}|^2}{X_{i,i}^2 X_{j,j}^2} + \frac{|Y_{i,j}|^2}{Y_{i,i}^2 Y_{j,j}^2} +\underbrace{ \sum_{i \neq j}^m \frac{1}{X_{i,i} X_{j,j}}\sum_{k=2}^\infty \left( \frac{|X_{i,j}|^2}{X_{i,i} X_{j,j}} \right)^k + \sum_{i \neq j}^m \frac{1}{Y_{i,i} Y_{j,j}}\sum_{k=2}^\infty \left( \frac{|Y_{i,j}|^2}{Y_{i,i} Y_{j,j}} \right)^k }_{\equiv\epsilon_1} \nonumber \\
    &+ 2\sum_{i, j=1}^m \frac{|W_{i,j}|^2}{X_{i,i}^2 Y_{j,j}^2}  + \underbrace{2\sum_{i, j=1}^m \frac{1}{X_{i,i} Y_{j,j}}\sum_{k=2}^\infty \left( \frac{|W_{i,j}|^2}{X_{i,i} Y_{j,j}} \right)^k}_{\equiv \epsilon_2}. \nonumber
}
In the last equation we introduced the errors $\epsilon_1, \epsilon_2 \geq 0$.
We can then write the variance as
\eq{
    \Delta^2 [d+e] =& \sum_{i=1}^m \left(\frac{1}{X_{i,i}}+\frac{1}{Y_{i,i}}-\frac{2}{X_{i,i}^2} -\frac{2}{Y_{i,i}^2} \right) +  \sum_{i ,j=1}^m \left( \frac{|X_{i,j}|^2}{X_{i,i}^2 X_{j,j}^2} + \frac{|Y_{i,j}|^2}{Y_{i,i}^2 Y_{j,j}^2} + 2 \frac{|W_{i,j}|^2}{X_{i,i}^2 Y_{j,j}^2} \right) +\epsilon,
}
where $\epsilon = \epsilon_1+\epsilon_2$. At this point we replace the diagonal elements $X_{i,i}, Y_{j,j}$ by their expectation $1+\mu$ and use the expressions for the Frobenius norms of $X$, $Y$ and $W$ to find
\eq{
\mathbb{E}\left(\Delta^2 [d+e] \right) \approx 2m \frac{(2-\mu) \mu}{(1-\mu) (1+\mu)^2},
}
where we ignore the small error $\mathbb{E}(\epsilon)$.

For the variance in the two halves we can easily write
\eq{
    \Delta^2d =& \sum_{i=1}^m \left(\frac{1}{X_{i,i}}-\frac{1}{X_{i,i}^2} \right) +  \sum_{i \neq j}^m \left(  \frac{1}{X_{i,i} X_{j,j} - |X_{i,j}|^2} - \frac{1}{X_{i,i} X_{j,j}} \right), \\
    \Delta^2e =& \sum_{i=1}^m  \left(\frac{1}{Y_{i,i}}-\frac{1}{Y_{i,i}^2} \right) +  \sum_{i \neq j}^m \left( \frac{1}{Y_{i,i} Y_{j,j} - |Y_{i,j}|^2} - \frac{1}{Y_{i,i} Y_{j,j}} \right),
}
and use the same approximations used previously to derive the results in Eq.~\eqref{app:single_cov}.
Finally, we can obtain the covariance between $d$ and $e$ and the variance of their difference by recalling that $\Delta^2[d \pm e] = \Delta^2 d + \Delta^2 e \pm \text{Cov}(d,e)$.


\section{Going beyond the dilute limit}
\label{sec:beyond_dilute}

The purpose of this section is to give evidence that \Cref{thm:main} can be extended beyond the dilute limit to show:
$\RGPEpa \in \FBPP^{\NP^\mathcal{O}}$
for some set of collision patterns that have high probability in the BipartiteGBS distribution. To do this, we first claim (without proof)\footnote{This may at first seem counterintuitive since a classical sampler may choose to favor certain collision patterns over others. However, the critical point is that the classical sampler does not know which collision pattern to favor since it will be sampled randomly.} that a generalization of  \Cref{thm:additive_estimation} exists where the subspace $\mathcal H_S$ can be replaced by any subspace $\mathcal H$ for which $S \in \mathcal H$ implies $S_{\pi,\sigma} \in \mathcal H$ for all $\pi, \sigma \in \mathrm S_m$.  Specifically, let us choose $\mathcal H$ to be the space of repetition patterns with $\mathbb E[\braket{n}]$ photon pairs and $\mathbb E[\braket{c}]/2$ clicks in each half of the modes (see Eq.~\eqref{mean_scaling} and Eq.~\eqref{app:means}, respectively). For this section, we will simply write these quantities as $n$ and $c$, respectively.  Therefore, we have
\beq
|\mathcal H| = \binom{m}{c}^2 \binom{n-1}{n-c}^2.
\eeq
Then, the generalization of \Cref{cor:perapprox} allows us to obtain a
\beq
\varepsilon \frac{\norm \alpha^{2n} m^n \prod_{i=1}^m s_i!\prod_{j=1}^m t_j! }{|{\mathcal H}|}
\eeq
approximation to $|\Per(A_{S,T})|^2$ where the pattern $(S,T)$ was chosen uniformly from $\mathcal H$.  On the other hand, to solve $\RGPEpa$, we need an $(\epsilon n! \prod_{i=1}^m s_i!\prod_{j=1}^m t_j!)$-additive approximation.  Following the same logic from \Cref{thm:main}, to obtain an $\FBPP^{\NP^\mathcal{O}}$ algorithm, we want
\beq
\label{eq:I_collision}
\mathcal I := \epsilon/\varepsilon
= \frac{\norm \alpha^{2n} m^n \prod_{i=1}^m s_i!\prod_{j=1}^m t_j!}{|{\mathcal H}|n!\prod_{i=1}^m s_i!\prod_{j=1}^m t_j!}
= \frac{\norm \alpha^{2n} m^n}{|{\mathcal H}|n!}.
\eeq
to be polynomially bounded.  Unfortunately, we can no longer use \Cref{lem:boundZ} to bound this quantity.  Because $\alpha = o(m^{1/4})$, the subdominant term of that expression ($e^{272m/\alpha^4}$) would \emph{still} contribute exponentially to Eq.~\eqref{eq:I_collision}.  Since that expression was obtained through several approximations, it is no longer tight enough.  Nevertheless, we claim that the empirical scaling of Eq.~\eqref{eq:I_collision} still appears to be polynomial.

Let us now describe how those numerics were obtained.  Given that we do not have an explicit bound on $\norm$, we must choose to estimate it somehow.  One approach would be to simply draw a random Gaussian matrix and compute $\norm$ from the eigenvalues of that matrix.  Unfortunately, this approach is quite slow in practice and we would only be able to calculate $\norm$ for relatively modest numbers of modes.  Instead, we will choose to estimate $\norm$ by an explicit formula.  First, consider the quantity
\beq
\norm^{-1} = \det(\mathbb{I}_m - A) = \prod_{i=1}^{m} (1- \lambda_i)
\eeq
for complex Wishart matrix $A \sim \mathcal W$ (see \Cref{app:bounds} for explanation of this ensemble).  In other words, the \emph{inverse} of $\norm$ is the characteristic polynomial of $A$.  The expectation of the characteristic polynomial for a Wishart matrix is well-understood:

\begin{lemma}[Edelman \cite{edelman1988eigenvalues}]
\label{lem:wishart_characteristic}
$\mathbb E[\norm^{-1}] = m! \sum_{i = 0}^m \binom{m}{i} \frac{(-1)^i}{(m-i)! (\alpha^2 m)^i}$.
\end{lemma}

In \Cref{fig:Zcalibration}, we show that the quantity $1/\mathbb E[\norm^{-1}]$ serves as a good estimate for $\norm$ by comparing it against its true value (computed directly with the eigenvalues).  In \Cref{fig:ratio}, we compute $\mathcal I$ from Eq.~\eqref{eq:I_collision} using the estimate $1/\mathbb E[\norm^{-1}]$ for $\norm$ and find a relatively benign scaling.  We conclude by reiterating that if we could \emph{prove} the accuracy of our estimator and that this scaling is only polynomial, then we arrive at our desired conclusion: $\RGPEpa \in \FBPP^{\NP^\mathcal{O}}$.  Therefore, if one were to also assume that $\RGPEpa$ is $\#\P$-hard, we obtain approximate average-case hardness for GBS in the high-collision $m = \Theta(\mathbb E[\braket{n}])$ regime.

\begin{figure}
\centering
    \input{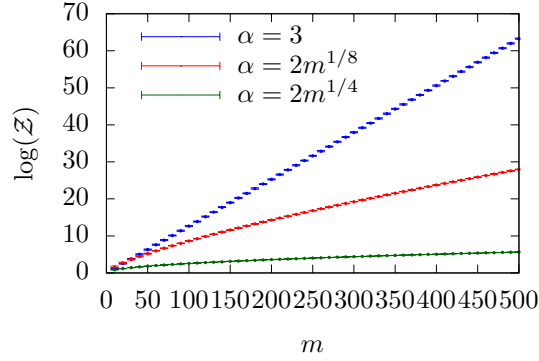}
    \caption{Comparison of $1/\mathbb E[\norm^{-1}]$ and $\norm$ as a function of $m$ and various $\alpha$ (shown on a log scale).  Each point was calculated from \Cref{lem:wishart_characteristic}, and the error bars show the sample standard deviation of $\log(\norm)$ calculated from 50 complex Wishart matrices drawn from $\mathcal W(0, \frac{1}{\alpha^2 m} \mathbb{I}_m)$.}
    \label{fig:Zcalibration}
\end{figure}

\begin{figure}
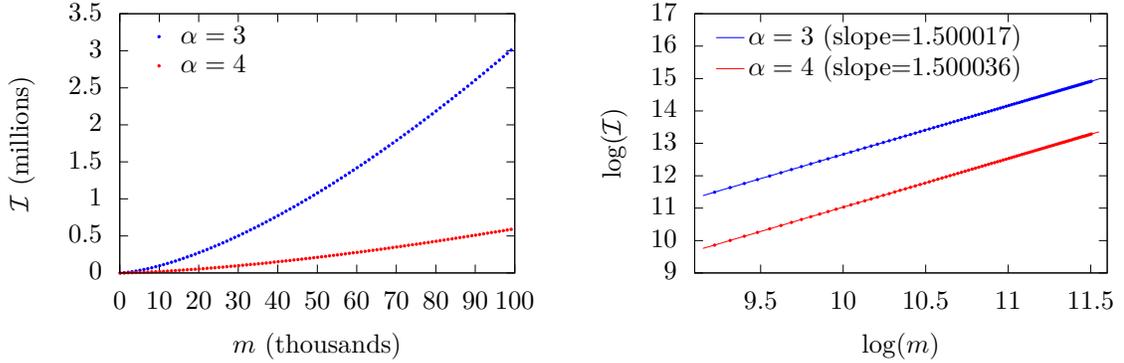
\centering
\begin{minipage}{.4\textwidth}
\input{figures/ratio.tex}
\end{minipage} \hspace{40pt}
\begin{minipage}{.4\textwidth}
\input{figures/ratio_loglog.tex}
\end{minipage}
\caption{\textit{Left}:  $\mathcal I := \frac{\norm \alpha^{2n} m^n}{|{\mathcal H}|n!}$ as a function of $m$ and various $\alpha$ (domain and range have been scaled by $10^3$ and $10^6$, respectively).  Large values of $\alpha$, such as those corresponding to the dilute limit are dwarfed by the scaling of $\mathcal I$ at these small values of $\alpha$, and so are not shown. \textit{Right}:  The same data as the left graph, except shown on a log-log plot and only values of $m$ greater than $10^4$.  Linear regression lines suggest that $\mathcal I$ scales like $m^{3/2}$ for small constant $\alpha$.}
\label{fig:ratio}
\end{figure}

\subsection{Justification for collision subspace}
\label{sec:H_justification}
The purpose of this section is to justify that the subspace $\mathcal H$ only contains collision patterns that suffice for the reduction given by \Cref{lemma:PerAsPoly}.  While it is known that \Cref{lemma:PerAsPoly} is insufficient to prove an efficient reduction between $\GPEa$ and $\RGPEa$, it is still weak evidence that the permanents that appear in the probabilities are hard to estimate.  To recap that procedure, we embed the permanent of matrix $A$ into the permanent of a matrix $B_{S,T}$ with repeated rows/columns provided the repetition patterns $S$ and $T$ have a least $k_T$ and $k_S$ modes (respectively) with a single photon.\footnote{Recall that while $k_T$ and $k_S$ are the number of collisions, they are also the number of modes that we need to have without collisions. This is due to the fact that every repeated row/column requires an \emph{unrepeated} column/row in the reduction.}  Notice that $\mathcal H$ contains repetition patterns with $k_T = k_S = n - c$.  So, to match the reduction, $n - c$ of the $c$ modes that click must have exactly one photon.  We will show that under a relatively modest condition on $\alpha$, all repetition patterns in $\mathcal H$ will satisfy this property.

While this may seem like a rather strong condition, the fact that the number of modes that click is always a significant fraction of the total number of photons makes it possible.  First notice that there are at most $n-c$ modes (of our total budget of $c$ clicks) that can have more than two photons.  For the reduction, we then need a different set of $n-c$ modes to click.  That is, we want
\beq \label{eq:c_is_big}
2(n-c) \le c.
\eeq
Using Eq.~\eqref{mean_scaling} and Eq.~\eqref{app:means}, this will hold whenever $\alpha \ge 3/\sqrt{2}$.

Finally, we claim that $A$ will be a $(2c -n) \times (2c - n)$ matrix because the total number of clicks must be the dimension of $A$ plus the $n-c$ extra clicks required for the reduction (i.e., $c = (2c - n) + (n-c)$).  In order to obtain a hardness proof, we would like $A$ to be relatively large so that computing its permanent is $\#\P$-hard.   We have that the size of $A$ is at least half the total number of clicks in the distribution whenever $2c - n \ge c/2$.  Notice that this is the same condition as that in Eq.~\eqref{eq:c_is_big}, and so we only require $\alpha \ge 3/\sqrt{2}$.  This is more than sufficient to guarantee hardness for all regimes from the dilute limit all the way to when the number of photons is linear in the number of modes.  We stress again, however, that our reduction is only efficient given an oracle for $\Per(B_{S,T})$, and the numerics of the previous section only suggest an oracle for $|\Per(B_{S,T})|^2$.
\end{document}